\documentclass[final]{IEEEtran}
\IEEEoverridecommandlockouts
\usepackage{color}
\usepackage{amsmath}
\usepackage{amsfonts}
\usepackage{amssymb}
\usepackage{amscd}
\usepackage[dvips]{graphicx}
\usepackage{epsfig}
\usepackage{newlfont}
\usepackage{cite}
\usepackage{subfigure}
\newtheorem{theorem}{Theorem}

\newcommand{\lb}{\left(}
\newcommand{\rb}{\right)}
\newcommand{\lsqb}{\left[}
\newcommand{\rsqb}{\right]}

\newcommand{\xbold}{\mathbf{x}}
\newcommand{\ybold}{\mathbf{y}}

\newcommand{\sbold}{\mathbf{s}}
\newcommand{\zbold}{\mathbf{z}}

\newcommand{\coopsignaltwo}{\mathbf{v}_{21}}
\newcommand{\vtwoone}{\mathbf{v}_{21}}

\newcommand{\SNRt}{\text{SNR}}
\newcommand{\INRt}{\text{INR}}

\ifCLASSINFOpdf
\else
\fi

\normalsize

\begin{document}
\title{Outer Bounds on the Secrecy Capacity Region of the $2$-user Z Interference Channel With Unidirectional Transmitter Cooperation}
\author{\authorblockN{Parthajit~Mohapatra$^{*}$,~Chandra R. Murthy$^{\ddag}$, and ~Jemin~Lee$^*$}\\
	\authorblockA{$^*$iTrust, Centre for Research in Cyber Security, Singapore University of Technology and Design, Singapore\\
		$^\ddag$Department of ECE, Indian Institute of Science, Bangalore\\
		Email: \{parthajit,~jemin\_lee\}@sutd.edu.sg, cmurthy@ece.iisc.ernet.in}}
\maketitle
\begin{abstract}
This paper derives outer bounds on the secrecy capacity region of the $2$-user Z interference channel (Z-IC) with rate-limited unidirectional cooperation between the transmitters. First, the model is studied under the linear deterministic setting. The derivation of the outer bounds on the secrecy capacity region involves careful selection of the side information to be provided to the receivers and using the secrecy constraints at the receivers in a judicious manner. To this end, a novel partitioning of the encoded messages and outputs is proposed for the deterministic model based on the strength of interference and signal. The obtained outer bounds are shown to be tight using the achievable scheme derived by the authors in a previous work. Using the insight obtained from the deterministic model, outer bounds on the secrecy capacity region of the $2$-user Gaussian Z-IC are obtained. The equivalence between the outer bounds for both the models is also established. It is also shown that secrecy constraint at the receiver does not hurt the capacity region of the $2$-user Z-IC for the deterministic model in the weak/moderate interference regime. On the other hand, the outer bounds developed for the Gaussian case shows that secrecy constraint at the receiver can reduce the capacity region for the weak/moderate interference regime. The study of the relative performance of these bounds reveals insight into the fundamental limits of the $2$-user Z-IC with limited rate transmitter cooperation.
\end{abstract}
\section{Introduction}
Interference is one of the primary factors in limiting the performance of wireless communication systems. Users are also susceptible to eavesdropping, due to the broadcast nature of the wireless medium. One way to tackle both these issues is through cooperation between the legitimate users. However, the effect of cooperation on secure communication in interference limited scenarios is not well understood. Such cooperation can affect the performance limits of the system in a completely different way compared to communication systems where reliable communication is the sole aim \cite{wang-TIT-2011, wang2, partha-arxiv-2014, partha-spawc-2013}. In this work, outer bounds on the secrecy capacity region are developed for the $2$-user Z interference channel (Z-IC) with rate-limited unidirectional cooperation between the transmitters.

An important information theoretic channel model, and the one investigated in this paper, is the Z-IC with unidirectional, rate-limited transmitter cooperation and secrecy constraints at the receivers~\cite{liu-globecom-2004,liu-icst-2011}.  Practically, it models, for example, a 2-tier network, where the macro cell user is close to the edge of the femtocell  while the femtocell user is close to the femto base station (BS). Since the macro BS can typically support higher complexity transmission schemes, it could use the side information received from the femto BS to precode its data to improve its own rate and simultaneously ensure secrecy at the femtocell user. At the receivers, the macro cell user could experience significant interference from the femtocell BS, while the femtocell user sees little or no interference from the macro BS, leading to the Z channel as the appropriate model for the system. Thus, the developed bounds give useful insights on the fundamental limits of communication. 
\subsection{Prior works}
The capacity region of the IC has remained an open problem, even without secrecy constraints at the receivers, except for some specific cases \cite{carleial-TIT-1975, sato-TIT-1981}. In \cite{carleial-TIT-1975}, it is shown that rate as high as that achievable without the interference can be achieved in $2$-user IC, when the interfering links are sufficiently stronger than the direct links. In \cite{sato-TIT-1981}, the capacity region of the IC is characterized for the strong interference regime. The IC with secrecy constraints has been analyzed in \cite{liu-TIT-2008,  lgamal2-TIT-2011}.  In \cite{tang-TIT-2011}, it is shown that a nonzero secrecy rate can be achieved even when the eavesdropper has a better channel compared to the legitimate receiver, in case of the wiretap channel with a helping interferer. This work also proposed computable outer bounds on the secrecy rate, where the tightest outer bound depends on the channel conditions.

It has been shown that limited rate cooperation between users can
improve the rates significantly in case of IC~\cite{wang-TIT-2011, wang2, vinod1}, when reliable communication is the sole aim. In \cite{wang-TIT-2011} and  \cite{vinod1} outer bounds on the capacity region of the $2$-user Gaussian IC are obtained when the transmitters can cooperate through a lossless link of finite capacity and a noisy link, respectively. In \cite{wang2}, outer bounds on the capacity region of the $2$-user Gaussian IC are obtained when the receivers can cooperate through a lossless link of finite capacity. However, obtaining outer bounds on the capacity region with transmitter cooperation is harder than obtaining bounds under receiver cooperation, because, in the former case, the encoded messages are dependent due to transmitter cooperation. In \cite{partha-arxiv-2014, partha-ncc-2014}, outer bounds on the secrecy capacity region of the $2$-user linear deterministic model and Gaussian symmetric IC are developed, where transmitters can cooperate through a lossless link of finite capacity. In \cite{geng-globecom-2016}, a tighter outer bound is obtained for the linear deterministic model with secrecy constraint and without transmitter cooperation for the initial part of the moderate interference regime.  Outer bounds on the secrecy capacity region for other communication models under different assumptions of cooperation can be found in \cite{ekrem2, ekrem1, awan1, peng-IFS-2013}.

The Z-IC model has also been studied in the existing literature with and without secrecy
constraints \cite{liu-globecom-2004, liu-TIT-2009, li-isit-2008}. In \cite{liu-globecom-2004},  outer bounds on the capacity region of
the Gaussian Z-IC for the weak/moderate interference regimes are derived when there is no secrecy constraints at the receivers. In
\cite{liu-TIT-2009}, for a special class of Z-IC,  the capacity region is established. In \cite{li-isit-2008}, for the
weak/moderate interference regime, the outer bounds on the secrecy capacity region of the $2$-user Z-IC without transmitter cooperation is shown to be tight
for the deterministic case. The outer bounds on the capacity region of the Z-IC with transmitter/receiver cooperation and without the secrecy constraint have been obtained in~\cite{bagheri-arxiv-2010, lei-TIT-2012, do-allerton-2009}. In \cite{bagheri-arxiv-2010}, both the encoders can cooperate through noiseless
links with finite capacities and the outer bounds developed helps to establish the sum capacity of the channel within $2$ bits per channel use. The role of cooperation between
the receivers is investigated in \cite{lei-TIT-2012, do-allerton-2009}. However, outer bounds on the secrecy capacity region of the $2$-user Z-IC with unidirectional transmitter cooperation have not been addressed in the existing literature. Deriving such bounds can offer key insight into the fundamental limits of secure communication in the Z-IC with unidirectional limited rate transmitter cooperation, and is therefore the focus of this work.
\subsection{Contributions}
This work considers unidirectional transmitter cooperation in the form of a rate-limited lossless link from transmitter~$2$ (which causes interference) to transmitter~$1$ (which does not cause
interference), and with secrecy constraints at receivers.  The objective of this paper is to derive outer bounds on the secrecy capacity region of the $2$-user Z-IC with unidirectional transmitter cooperation and secrecy constraints at the receivers. This, in turn, requires judicious use of the secrecy constraint at receiver, along with careful selection of the side information to be provided to the receivers. In particular,  the cooperation between the transmitters makes the encoded messages dependent, which makes derivation of the outer bounds even more difficult.

First, the problem is solved under the deterministic approximation of the channel. The study of the deterministic model gives useful insights, and motivates the
outer bounds in the Gaussian setting. However, it is non-trivial to extend the results obtained for the deterministic case to the Gaussian setting due to the well known differences between the two models \cite{avesti1}. The main contributions of the paper are summarized below:
\begin{enumerate}
 \item The key novelty in deriving outer bounds on the secrecy capacity region for the deterministic model is the choice of side information to be provided to the receiver(s) and the judicious use of the secrecy constraints at the receivers. To elaborate, a novel partitioning of the encoded messages and outputs is proposed for the deterministic model based on the strength of interference and signal. This partitioning helps to bound or simplify the entropy terms that are difficult to evaluate due to the dependence between the encoded messages.

 \item Outer bounds are developed for the Gaussian case by providing appropriate side information and bounding the entropy terms containing both discrete and continuous random variables, based on the insights obtained for the deterministic case (Sec.~\ref{sec:outer-gaussian}).  The outer bounds derived on the secrecy capacity region of the Gaussian Z-IC are the best known outer bounds till date with unidirectional transmitter cooperation.
  
  \item  The outer bounds on the secrecy capacity region of the $2$-user Z-IC without
  cooperation between the transmitters can be obtained as special case of the analysis for both the
  models. Note that, prior to this work, the capacity region of the Z-IC  for the deterministic model with secrecy constraints was
  not fully known even for the non-cooperating case~\cite{li-isit-2008}.
\end{enumerate}

The outer bounds on the secrecy capacity region of the $2$-user Z-IC for the deterministic model does not use the secrecy constraints at the receivers in the weak/moderate interference regime. It is found that the achievable results obtained for this model in \cite{partha-isit-2015} matches with the outer bounds derived in this work for the weak/moderate interference regime. Hence, there is no penalty on the capacity region of the Z-IC due to the secrecy constraints at the receivers in the weak/moderate interference regimes.  However, in the very high interference regime, irrespective of the capacity of the cooperative link, the outer bounds developed for the deterministic model shows that user~$2$ cannot achieve any nonzero secrecy rate. On the other hand, the outer bounds developed for the Gaussian case show that secrecy constraint can reduce the capacity region of the Z-IC in all the interference regimes. Part of this work has appeared in \cite{partha-isit-2015}.

\textit{Notation:} Lower case or upper case letters represent scalars, lower case
boldface letters represent vectors, upper case boldface letters represent
matrices, $(x)^{+} \triangleq \max\{0 ,x\}$ and $\lfloor . \rfloor$ denotes the floor operation. 

\textit{Organization:} Section~\ref{sec:system_model} presents the system model.
In Secs.~\ref{sec:outer-dermin}~and~\ref{sec:outer-gaussian}, the outer bounds for the deterministic and Gaussian models are presented, respectively.
In Sec.~\ref{sec:results}, some numerical examples are presented to offer a
deeper insight into the bounds. Concluding remarks are offered in
Sec.~\ref{sec:conclusion}; and the proofs of the theorems are provided in the Appendices.
\section{System Model}\label{sec:system_model}
\begin{figure}[t]
\centering
\mbox{\subfigure[Gaussian model]{\includegraphics[width=1.6in, height = 1.6in]{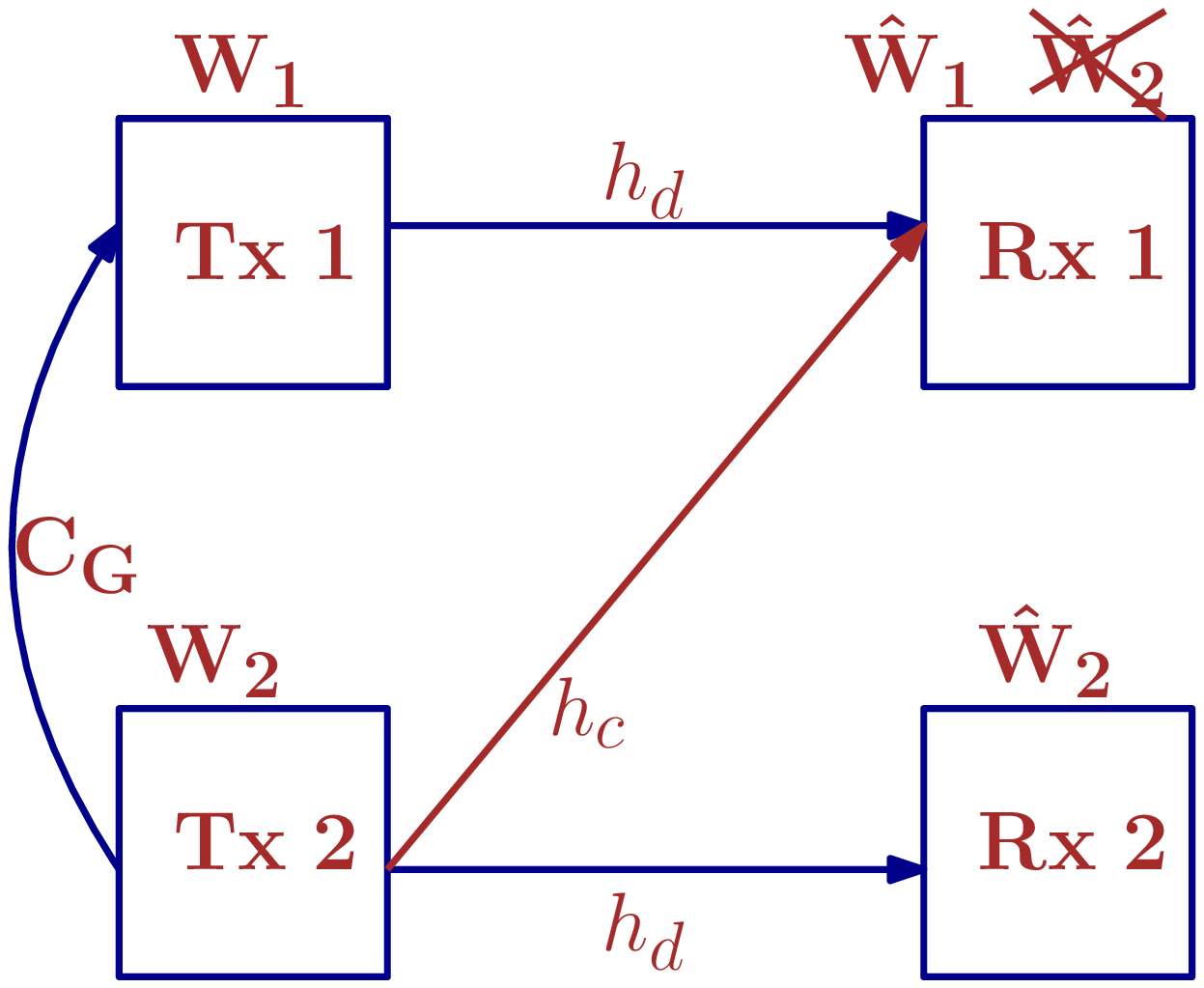} \label{fig:gaussian_model}}\quad
\subfigure[Deterministic model]{\includegraphics[width=1.6in, height= 1.6in]{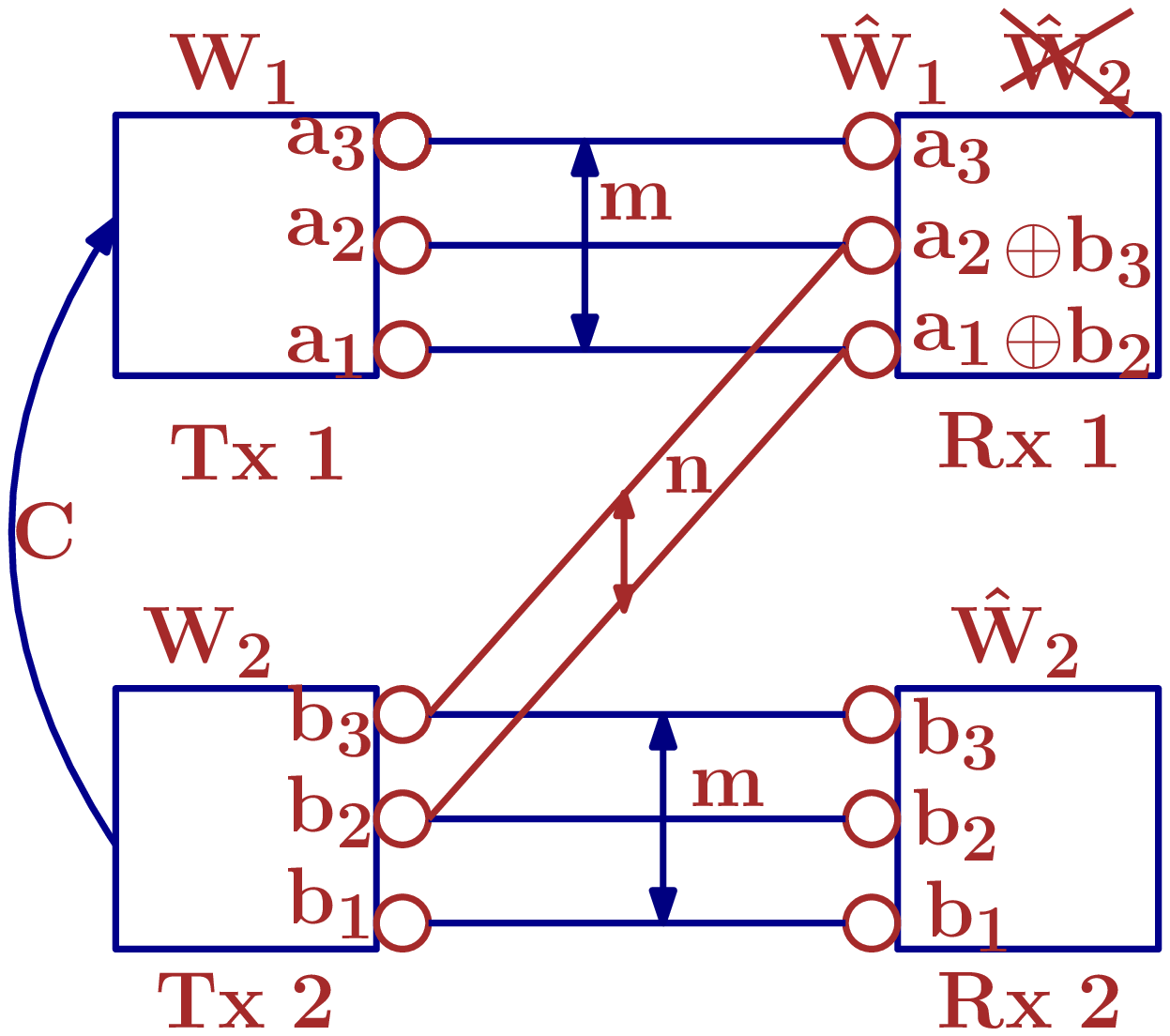} \label{fig:deterministic_model}}}
\caption[]{$2$-user Z-IC with unidirectional transmitter cooperation (from transmitter~$2$ to transmitter~$1$).}\label{fig:outersplit}
\end{figure}
Consider a $2$-user Gaussian symmetric Z-IC  with unidirectional and rate-limited transmitter
cooperation from transmitter~$2$ to $1$, as shown in Fig.~\ref{fig:gaussian_model}.\footnote{The model is
termed as symmetric since the links from transmitter~$1$ to receiver~$1$  and from transmitter~$2$ to receiver~$2$
are of the same strength.} In the Z-IC, only one of the users (i.e., transmitter~$2$) causes interference to the unintended receiver (i.e., receiver~$1$). The received signal at
receiver~$i$, $\ybold_i$, is given by
\begin{align}
y_{1} = h_d x_1 + h_c x_{2} + z_1;  y_2 = h_dx_2 + z_2, \label{sysmodel1}
\end{align}
where $z_j$ $(j=1,2)$ is the additive white Gaussian noise, distributed as $\mathcal{N}(0,1)$. Here, $h_d$ and $h_c$ are the channel gains
of the direct and interfering links, respectively. The input signals ($x_i$)
are required to satisfy the power constraint: $E[|x_i|^2] \leq P$.  The transmitter~$2$ cooperates with transmitter~$1$ through a noiseless and secure link of finite rate denoted by $C_G$.

The equivalent deterministic model of \eqref{sysmodel1} at high SNR is given by~\cite{li-isit-2008,wang-TIT-2011}
\begin{align}
\mathbf{y}_{1} = \mathbf{D}^{q-m}\mathbf{x}_{1} \oplus \mathbf{D}^{q-n}\mathbf{x}_{2}; \quad  \mathbf{y}_{2} = \mathbf{D}^{q-m}\mathbf{x}_{2}, \label{sysmodel2}
\end{align}
where $\mathbf{x}_{1}$ ($\mathbf{x}_{2}$) is the binary input vector of the deterministic Z-IC from user~$1$
(user~$2$) of length $m$ ($\max\{m,n\}$); $\mathbf{y}_{1}$ $(\mathbf{y}_{2})$ is the binary output vector of
length $\max\{m,n\}$ ($m$);  $\mathbf{D}$ is a $q \times q$ downshift matrix with elements $d_{j',j''}=1$ if
$2 \leq j'=j''+1\leq q$ and $d_{j',j''}=0$ otherwise; and the operator $\oplus $ stands for  modulo-$2$ addition,
i.e., the \textsf{XOR} operation. The deterministic model is also shown in Fig.~\ref{fig:deterministic_model}.

The deterministic model is a first order approximation of a Gaussian channel, where all the signals are represented
by their binary expansions. Here, noise is modeled by truncation, and the superposition of signals at the
 receiver is modeled by \textit{modulo}~$2$ addition. Hence, the parameters $m$, $n$, and $C$ of the deterministic 
 model are related to the Gaussian symmetric Z-IC as $ m = (\lfloor 0.5 \log \text{\textsf{SNR}}\rfloor)^{+},\: n = (\lfloor 0.5 \log \text{\textsf{INR}}\rfloor)^{+},$ 
 and $C = \lfloor C_G\rfloor$. Note that the
 notation followed for the deterministic model is the same as that presented in \cite{wang-TIT-2011}. The
 bits $a_{i} \in \mathcal{F}_{2}$ and $b_i \in \mathcal{F}_{2}$ denote the  information bits of transmitters $1$ and $2$, respectively, sent
 on the $i^{\text{th}}$ level, with the levels numbered starting from the bottom-most entry. 

The transmitter~$i$ has a message $W_{i}$, which should be decodable at the intended receiver~$i$,
but needs to be kept secret from the other, i.e., the unintended receiver $j$ ($j \neq i$), and this is termed as the \emph{secrecy constraint}. Note that, for the Z-IC, the message $W_1$ is secure as there is no link from transmitter~$1$ to receiver~$2$. Hence, the goal is to ensure that $W_2$ is not decodable at receiver~$1$. The encoding at  
transmitter~$1$ should satisfy the causality constraint, i.e., it cannot depend 
on the signal to be sent over the cooperative link in the future. The signal sent over the cooperative link from
transmitter~$2$ to transmitter~$1$ is represented by $\mathbf{v}_{21}$. It is
also assumed that the transmitters trust each other completely and they do not deviate from the agreed
schemes, for both the models. For both the models, the encoded message at transmitter~$1$ is a function of its own message, the signal received over the cooperative link and possibly some artificial noise or jamming signal, whereas, the encoded message at transmitter~$2$ is independent of the other user's message.  In the derivation of the outer bounds for the deterministic and Gaussian models, the notion of weak secrecy is considered, i.e., $\frac{1}{N} I(W_2; \ybold_1^N) \rightarrow 0$ as $N \rightarrow \infty$, where $N$ corresponds to the block length~\cite{wyner-bell-1975}.

The following interference regimes are considered: weak/moderate interference regime $(0 \leq \alpha \leq 1)$, high interference regime $(1 < \alpha \leq 2)$
and very high interference regime $(\alpha > 2)$, where, with slight abuse of notation $\alpha \triangleq
\frac{n}{m}$ is used for the deterministic model and $\alpha \triangleq \frac{\log \INRt}{\log \SNRt}$ is used for the
Gaussian model. The quantity $\alpha$ captures the amount of coupling between the
signal and interference.
\section{Outer Bounds for the Linear Deterministic Z-IC Model}\label{sec:outer-dermin}
In this section, outer bounds on the secrecy capacity region for the linear deterministic Z-IC with
unidirectional transmitter cooperation are presented for the different interference regimes as 
Theorems~\ref{th:theorem-weakmod-outer1}-\ref{theorem-veryhigh-outer1}. Note that in all interference regimes, the rate of both the users can be trivially upper bounded by $m$, i.e., $R_1 \leq m$ and $R_2 \leq m$. One of the key techniques used in deriving tight outer bounds is to partition the encoded
 message, output, or both, depending on the value of $\alpha$.  The partitioning of the encoded messages/outputs gives insights on the side information to be provided to the receiver. This in turn allows one to exploit the secrecy constraint at the receiver to obtain tight and tractable outer bounds on the secrecy capacity region of the Z-IC. This partitioning also helps to simplify the entropy terms as the encoded messages at the transmitters are not independent due to the cooperation between the transmitters. 
 
 The following Markov relation is used in the derivation of these outer bounds: conditioned
 on the cooperating signal $(\coopsignaltwo^N)$, the encoded signals and the messages at the
 two transmitters are independent \cite{willems-TIT-1983, wang-TIT-2011}, i.e.,
\begin{align}
& (W_1, \xbold_1^N) \rightarrow (\vtwoone^N) \rightarrow (W_2,\xbold_2^N). \label{eq:usefulrelation}
\end{align}
For the derivation of the first outer bound  in the weak/moderate interference regime, the encoded message $\xbold_1$ is
partitioned into two parts: one part ($\xbold_{1a}$), which is received without interference at receiver~$1$, and
another part ($\xbold_{1b}$), which is received with interference at receiver~$1$. The encoded message
of transmitter~$2$ is also split into two parts: one part ($\xbold_{2a}$), which causes interference to receiver~$1$,
and another part ($\xbold_{2b}$), which does not cause any interference to receiver~$1$. 
The partitioning of the output and the encoded message is shown in Fig.~\ref{fig:weakmodouter}.  In the derivation of this outer bound, the secrecy constraints at the receivers are not used. To get insights on this, consider the following two cases in Fig.~\ref{fig:ach_scheme}. In Fig.~\ref{fig:wekmod_ach_case1}, it can be noticed that user~$1$ can transmit $m$ bits securely as there is no link from transmitter~$1$ to receiver~$2$. Hence, user~$1$ can achieve the maximum rate of $m$. On the other hand, user~$2$ can transmit on the lower levels $[1:m-n]$ and it can send $m-n$ bits securely. Hence, the rate point $(m, m-n)$ is achievable. In Fig.~\ref{fig:weakmod_ach_case2}, user~$2$ sends data bits on the levels $[1:m]$.  As the data bits sent on the levels $[m-n+1:m]$ are received at receiver~$1$, transmitter~$1$ sends random bits 
generated from $\mathcal{B}(\frac{1}{2})$ distribution on the levels $[1:n]$ to ensure secrecy of the data bits of transmitter~$2$ at receiver~$1$. Hence, user~$2$ can also achieve the maximum possible rate of $m$. On the remaining 
levels $[n+1:m]$, transmitter~$1$ sends its own data bits. As transmitter~$2$ does not cause any interference to the data bits sent on these levels by transmitter~$1$, receiver~$1$ can decode these data bits. Hence, transmitter~$1$ achieves a
rate of $m-n$, and the rate point $(m-n, m)$ is achievable. It is not difficult to see that, even if there is no secrecy constraint at the receivers, it is not possible to achieve rates exceeding the corner points $(m, m-n)$ and $(m-n, m)$. This motivates one to derive the outer bounds on the capacity region without using the secrecy constraints at the receivers.
\begin{figure}[t]
\centering
\mbox{\subfigure[][$(m,n)=(5,3)$]{\includegraphics[width=1.6in,height=1.7in]{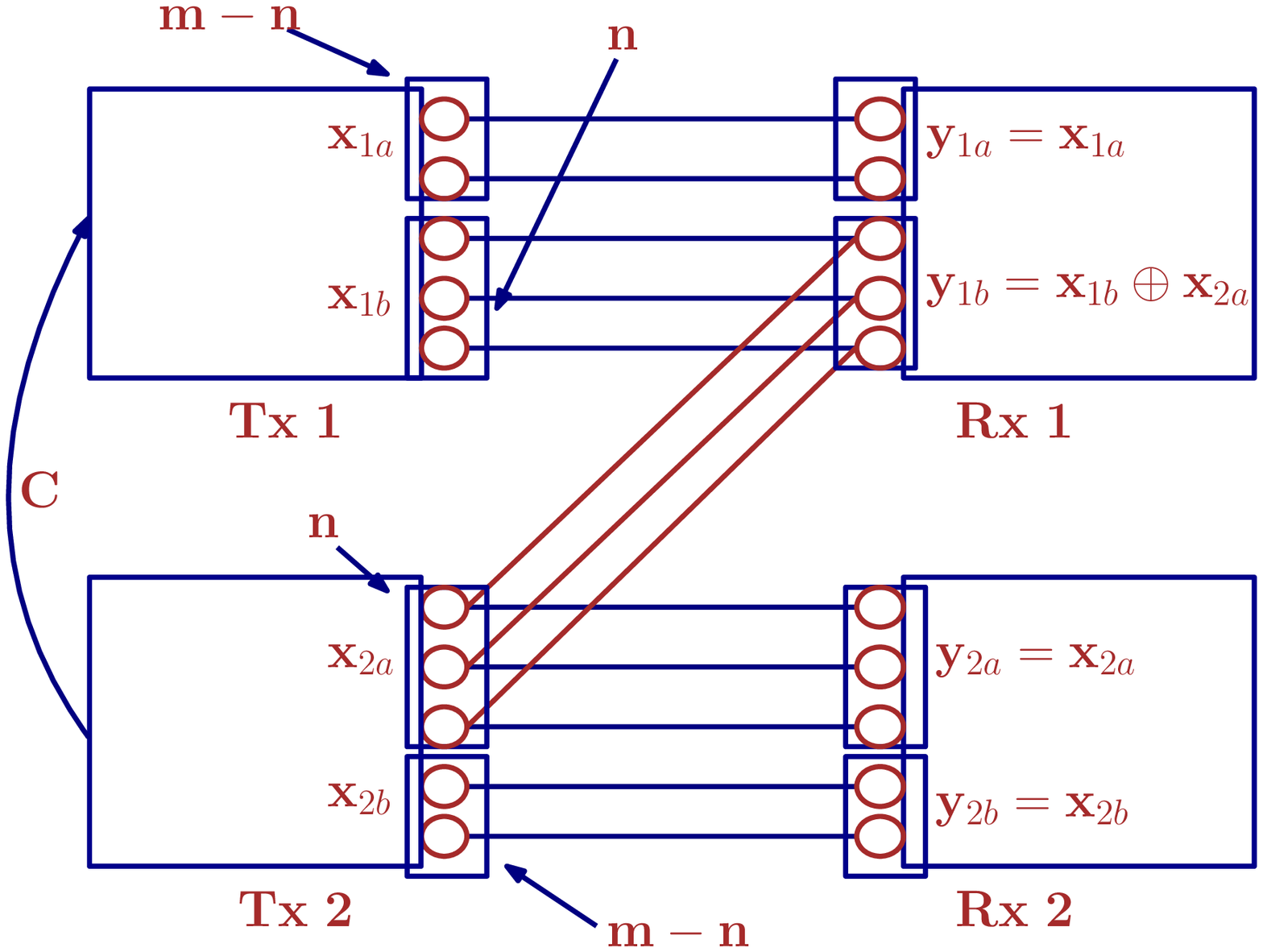}\label{fig:weakmodouter}} \quad
\subfigure[][$(m,n)=(4,5)$]{\includegraphics[width=1.6in,height=1.7in]{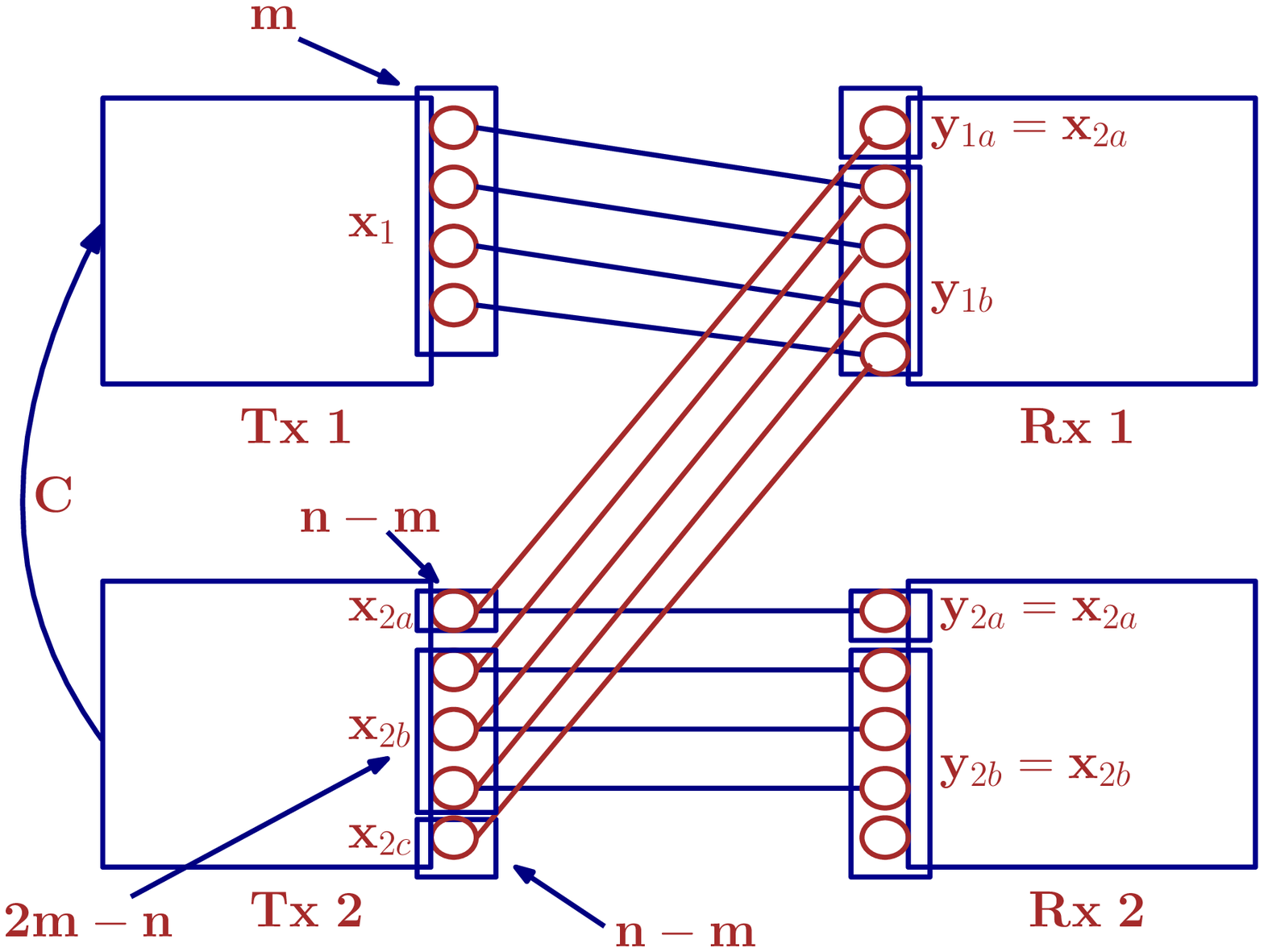}\label{fig:highouter}}} \\
\caption[]{Deterministic model: partitioning of encoded messages and outputs.}\label{fig:outerpartition}
\end{figure}
\begin{theorem}\label{th:theorem-weakmod-outer1}
In the weak and moderate interference regimes, i.e., $0 \leq \alpha \leq 1$, the secrecy capacity region of the $2$-user
deterministic Z-IC with unidirectional transmitter cooperation is upper bounded as
\begin{align}
& R_1 \leq m, R_2 \leq m, \text{ and } \nonumber \\
& R_1 + R_2 \leq 2m-n + C. \label{eq:th-weakmod-outer1}
\end{align}
\end{theorem}
\begin{proof}
See Appendix~\ref{sec:append-det-ZIC-outer1}.
\end{proof}
\textit{Remarks:}
\begin{itemize}
\item Interestingly, using the above theorem and the achievable result in \cite{partha-isit-2015}, it can be shown that the secrecy constraints at the receivers 
\emph{do not result in any penalty} on the capacity region. Thus, secrecy can 
be obtained for free in the weak/moderate interference regimes, for all the 
values of $C$. The outer bound in Theorem~\ref{th:theorem-weakmod-outer1} also serves as outer bound on the capacity region of the $2$-user Z-IC with unidirectional limited rate transmitters cooperation, when there is no secrecy constraints at the receivers.
\item When $C=0$ and $0 < \alpha \leq 1$, 
the outer bound derived in Theorem~\ref{th:theorem-weakmod-outer1}
matches the outer bound in Theorem~$2$ in \cite{li-isit-2008}.
\end{itemize}
\begin{figure}[t]
	\centering
	\mbox{\subfigure[][$(R_1, R_2) = (5, 2)$]{\includegraphics[width=1.6in,height=1.5in]{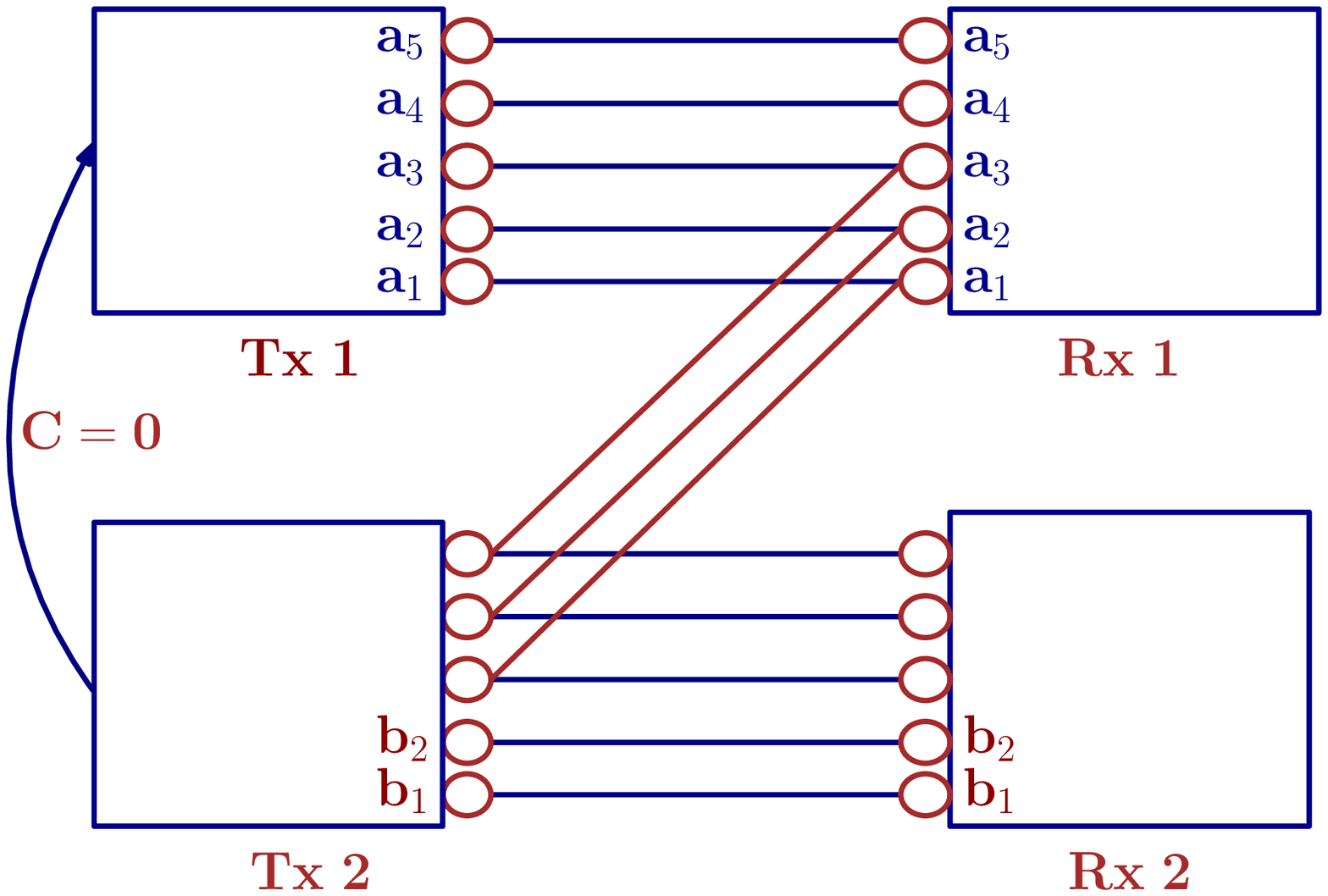}\label{fig:wekmod_ach_case1}} \quad
		\subfigure[][$(R_1, R_2) = (2, 5)$] {\includegraphics[width=1.6in,height=1.7in]{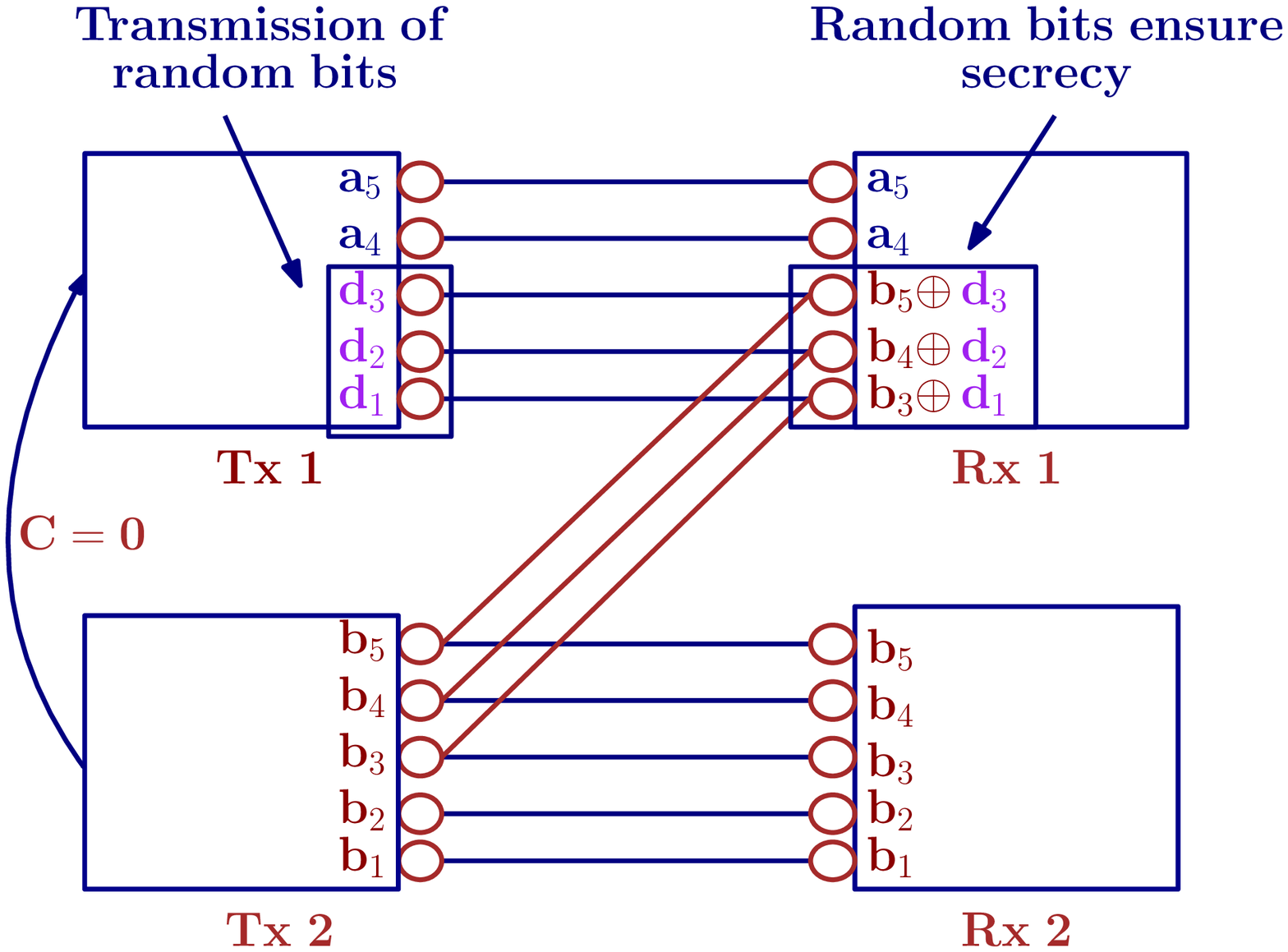}\label{fig:weakmod_ach_case2}}} \\
	\caption[]{Achievable schemes: deterministic Z-IC with $m =5, n=3$ and $C=0$.} \label{fig:ach_scheme}
\end{figure}
It is intuitive to think  that as the strength of interference increases, the achievable secrecy rate may decrease. In particular, in the high/very high interference regimes, the secrecy constraint may lead to a rate penalty, in contrast to the weak/moderate interference regime. Hence, in the high/very high interference regime, the secrecy constraint is used along with the reliability criteria in the derivation of the outer bounds.

First, consider the high interference regime, i.e., $1 < \alpha < 2$. In this case, it is not difficult to see that the rate of user~$1$ can be upper bounded by $m$. To get
insights into the outer bounds on $R_2$ and $R_1 + R_2$, consider Fig.~\ref{fig:highouter}. One can see that transmitter~$2$
cannot use the levels $[1:n-m]$ for transmitting its own data as the corresponding links do not exist at the
intended receiver. Any data bits transmitted on the levels $[m+1:n]$, i.e., $\xbold_{2a}$, will be received
without interference at receiver~$1$. If receiver~$2$ can decode these data bits, receiver~$1$ will also be
able to decode these data bits. Hence, these data bits $\ybold_{1a} = \xbold_{2a}$ will not be secure. Hence, they are provided as side information to receiver~$2$ to obtain the upper bounds. Then, using the secrecy constraint at receiver~$1$, the following outer bounds
can be obtained.
\begin{theorem}\label{theorem-high-outer1}
In the high interference regime, i.e., $1 < \alpha < 2$, the secrecy capacity region of the $2$-user deterministic Z-IC with
unidirectional transmitter cooperation is upper bounded as
\begin{align}
& R_1 \leq m, R_2 \leq 2m-n, \text{ and }\nonumber \\
& R_1 + R_2 \leq m+C. \label{eq:th-high-outer1}
\end{align}
\end{theorem}
\begin{proof}
See Appendix~\ref{sec:append-ZIC-outer2}.
\end{proof}
\textit{Remarks:} 
\begin{itemize}
\item The outer bound on $R_2$ shows that there is a nonzero penalty on the rate of user~$2$ due to the secrecy constraint at receiver~$1$, in contrast to the weak/moderate interference regime (see Theorem~\ref{th:theorem-weakmod-outer1}).
\item When $C=0$, the outer bound on the sum rate suggests that for user~$2$ to achieve a nonzero secrecy rate, user~$1$ has to sacrifice some of its rate.
\end{itemize}
\begin{figure}
\begin{center}
\includegraphics[width=2.5in,height=2.2in]{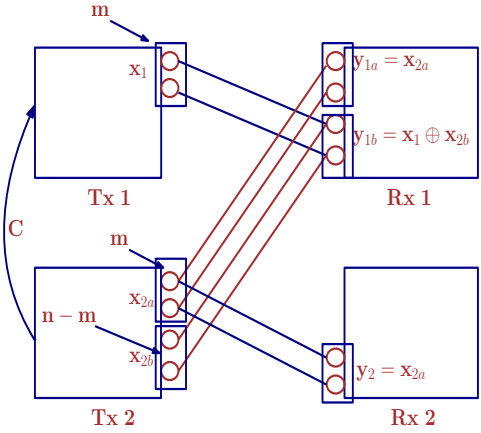}
\caption{Deterministic Z-IC with $(m,n)=(2,4)$: Illustration of partitioning of the message/output.}
 \label{fig:veryhighouter}
\end{center}
\end{figure}

Now, consider the derivation of the outer bound for the very high interference regime. In
Fig.~\ref{fig:veryhighouter},  it can be noticed that only the levels 
$[n-m+1:m]$ can be used to send data from transmitter~$2$ to receiver~$2$,
as the links corresponding to the lower levels $[1:n-m]$ do not exist at
receiver~$2$. The data bits transmitted on the levels $[n-m+1:n]$, i.e., 
$\xbold_{2a}$, 
are received without interference at receiver~$1$. If receiver~$2$ can
decode these data bits, then receiver~$1$ can also decode these data bits.
Hence, transmitter~$2$ cannot send any data bits securely on these levels. To capture this in
the derivation, receiver~$2$ is provided with the side information of the form $\ybold_{1a}^N$, which in
turn helps to bound the rate by $I(W_2;\ybold_2^N|\ybold_{1a}^N)$. It can be
noticed that this quantity is zero as $\ybold_{1a} = \ybold_2 = \xbold_{2a}$. The outer bounds on the secrecy capacity region in the very high interference regime $(\alpha \geq 2)$ are given in the following theorem. 
\begin{theorem}\label{theorem-veryhigh-outer1}
In the very high interference regime, i.e., $ \alpha \geq 2$,  the secrecy capacity of the
$2$-user Z-IC with unidirectional cooperation  is upper bounded by $R_1 \leq m$ and $R_2 \leq 0$.
\end{theorem}
\begin{proof}
See Appendix \ref{sec:append-ZIC-outer3}.
\end{proof}
\textit{Remark:}
Since $R_2$ is upper bounded by zero in the very high interference regime, it is not possible for transmitter~$2$ to 
achieve any nonzero secrecy rate, irrespective of $C$.


\section{Outer Bounds for the Gaussian Z-IC Model}\label{sec:outer-gaussian}
In this section, the outer bounds on the secrecy capacity region for the Z-IC with unidirectional transmitter cooperation are stated as Theorems~\ref{th:theorem-gauss-outer1}-\ref{th:theorem-gauss-outer2}. 
Although these outer bounds are related to the corresponding outer bounds derived in the deterministic case, the extension of the outer bounds to the Gaussian case is non-trivial due to the following differences between the two models. First, in the Gaussian case, noise cannot be modeled by truncation as in the deterministic case. Second, in the Gaussian case, the superposition of signals or interference is modeled by real addition in contrast to the modulo-2 addition used in the deterministic case. Hence, the carry over involved in the real addition is not captured in the deterministic case. Third, the derivation of the outer bound in the Gaussian case involves bounding of differential entropy terms containing continuous as well as discrete random variables, due to the unidirectional cooperation between the transmitters. This makes the derivation of the outer bounds even more difficult.   

Hence, the partitioning of the encoded messages or outputs used in the derivation of the outer bounds for the deterministic case cannot  be directly applicable to the Gaussian case. To overcome this problem, either analogous quantities that serve as side-information at receiver need to be found to mimic the partitioning of the encoded messages/outputs, or the bounding steps need to be modified taking cue from the deterministic model. This helps to obtain tractable outer bounds on the secrecy capacity region, which are presented in the following subsections.  
\subsection{Weak/moderate interference regime $(0 \leq \alpha \leq 1)$}
 The outer bound derived in Theorem~\ref{th:theorem-weakmod-outer1} involved providing  the side information 
 $(\xbold_{2a}, \vtwoone)$ to receiver~$2$ by a genie. The quantity $\xbold_{2a}$ corresponds to the part of the encoded message $\xbold_{2}$ of transmitter~$2$ which causes interference at receiver~$1$ (See Fig.~\ref{fig:weakmodouter}). In the Gaussian case, to mimic the approach 
used for the deterministic case, receiver~$2$ is provided with side information  $(\sbold_2 \triangleq h_c\xbold_2 + \zbold_1, \vtwoone)$. 
Note that outer bound based on this idea was presented in~\cite{bagheri-arxiv-2010}, which considered
Gaussian Z-IC with unidirectional transmitter cooperation, but without secrecy constraints at the receivers. For the sake of 
completeness, the result is stated as Theorem~\ref{th:theorem-gauss-outer1}. The outer bound in Theorem~\ref{th:theorem-weakmod-outer1} for the weak/moderate interference regime can be considered as deterministic equivalent of the outer bound presented below. 
\begin{theorem}[\cite{bagheri-arxiv-2010}]\label{th:theorem-gauss-outer1}
The capacity region of the 2-user Gaussian Z-IC with unidirectional transmitter cooperation is
 upper bounded as
\begin{align}
& R_1 \leq 0.5\log(1+\text{SNR}), \: R_2 \leq 0.5 \log(1+ \text{SNR}), \nonumber \\
& R_1 + R_2 \leq 0.5\log(1 + \text{SNR} + \text{INR} + 2\sqrt{\text{SNR}\cdot\text{INR}}) \nonumber \\
& \qquad \qquad + 0.5\log\lb1 + \frac{\text{SNR}}{1 + \text{INR}} \rb + C_G, \label{eq:zic-outer-gaussian1}
\end{align}
where $\text{SNR} \triangleq h_d^2P$ and $\text{INR} \triangleq h_c^2P$.
\end{theorem}
Note that the outer bound stated in Theorem~\ref{th:theorem-gauss-outer1} does not use the secrecy constraint at receiver. In the weak/moderate interference regime, the data bits transmitted on the lower levels $[1:m-n]$ of transmitter~$2$ are inherently secure in the deterministic case as shown in Fig.~\ref{fig:ach_scheme}. However, in the Gaussian case, there is no one-to-one analogue of this as noise cannot be modeled by truncation. The secrecy constraint at the receiver may lead to a nonzero penalty in rate for the Gaussian case. Hence,  outer bounds are derived on the rate of user~$2$ and sum rate using the 
secrecy constraint at receiver~$1$, which is stated as a theorem below.  
\begin{theorem}\label{th:theorem-gauss-outer1a}
The secrecy capacity region of the $2$-user
Gaussian Z-IC with unidirectional transmitter cooperation is upper bounded as
\begin{align}
R_1 & \leq 0.5\log(1 + \SNRt), \nonumber \\
R_2 & \leq \displaystyle\max_{- 1 \leq \rho \leq 1}0.5\log\Bigg( 1 + \text{SNR} \nonumber \\
&  \qquad - \frac{(\rho\text{SNR} + \sqrt{\text{SNR}\cdot\text{INR}})^2}{1+ \text{SNR}+\text{INR}+2\rho\sqrt{\text{SNR}\cdot\text{INR}}}\Bigg), \nonumber \\
R_1 + R_2 &\leq \log(1 + \SNRt) - 0.5\log(1 + \INRt) + C_G. \label{eq:theorem-gauss-outer1a}
\end{align}
\end{theorem}
\begin{proof}
See Appendix~\ref{sec:append-ZIC-Gaussian-outer1a}.
\end{proof}
\textit{Remarks:}
\begin{itemize}
\item It is easy to show that the outer bounds on the sum rate in Theorem~\ref{th:theorem-gauss-outer1a} is tighter than the outer bound in Theorem~\ref{th:theorem-gauss-outer1} for all values of SNR, INR and $C_G$. Thus, the outer bound in Theorem~\ref{th:theorem-gauss-outer1a} improves over 
Theorem~\ref{th:theorem-gauss-outer1}. 
From the outer bound on the rate of user~$2$ in Theorems~\ref{th:theorem-gauss-outer1} and \ref{th:theorem-gauss-outer1a}, it can be observed that outer bound obtained with secrecy constraint is tighter compared to the outer bound obtained without using the secrecy constraint. 
\item When $C_G =0$, the outer bound on the rate of user~$2$ reduces to $0.5\log \lb 1 + \SNRt - \frac{\SNRt \cdot \INRt}{1 + \SNRt + \INRt}\rb$, as the only possible value $\rho$ can take is $0$. Hence, this outer bound indicates that user~$2$ cannot achieve the maximum possible rate of $0.5\log \lb 1 + \SNRt\rb$. This is in contrast to the deterministic case, where user~$2$ can achieve the maximum rate of $m$, as observed from Theorem~\ref{th:theorem-weakmod-outer1} and Fig.~\ref{fig:ach_scheme}.
\item The outer bound on the sum rate in Theorem~\ref{th:theorem-gauss-outer1} is applicable for all the interference regimes whereas the outer bound in Theorem~\ref{th:theorem-gauss-outer1a} is applicable only in the weak/moderate interference regime. 
\end{itemize}
\subsection{High interference regime $(1 < \alpha < 2)$}
The derivation of the outer bound in this regime is based on the
outer bound in Theorem~\ref{theorem-high-outer1} obtained for the deterministic model. 
In the proof of Theorem~\ref{theorem-high-outer1}, to upper bound the rate of user~$2$, a part of the output at receiver~$1$ which does not contain signal sent by transmitter~$1$ is provided as side-information to receiver~$2$, i.e., $\ybold_{1a}^N$. In the Gaussian case, it is not 
possible to partition the encoded message as it was done for the deterministic model (See Fig.~\ref{fig:highouter}). To overcome this problem, output at receiver~$1$, i.e., $\ybold_1^N$, is provided as side information to receiver~$2$.  Providing side information in this way creates a degraded channel from transmitter~$2$ to receiver~$1$ with respect to the channel from transmitter~$2$ to receiver~$2$. In the deterministic case, to upper bound the sum rate, the output at receiver~$1$ $(\ybold_1^N)$ is partitioned into two parts: $\ybold_{1a}^N$ and  $\ybold_{1b}^N))$, and receiver~$2$ is provided with side information of the form $\ybold_{1a}^N$. To mimic this in the Gaussian case, output of receiver~$2$, i.e., $\ybold_2^N$, is provided as side information to receiver~$1$ and $(W_1, \ybold_1^N)$ is provided as side information to receiver~$2$.  The outer bound on the secrecy capacity region is stated in the following theorem.
\begin{theorem}\label{th:theorem-gauss-outer2}
The secrecy capacity region of the 2-user Gaussian Z-IC with unidirectional transmitter cooperation is upper
bounded as
\begin{align}
R_1 & \leq 0.5\log(1+\text{SNR}), \nonumber \\
R_2 & \leq \displaystyle\max_{- 1 \leq \rho \leq 1}0.5\log\Bigg( 1 + \text{SNR} \nonumber \\
& \quad - \frac{(\rho\text{SNR} + \sqrt{\text{SNR}\cdot\text{INR}})^2}{1+ \text{SNR}+\text{INR}+2\rho\sqrt{\text{SNR}\cdot\text{INR}}}\Bigg), \nonumber 
\end{align}
\begin{align}
R_1 + R_2 &\leq \displaystyle\max_{- 1 \leq  \rho \leq 1} 0.5\log\Bigg( 1 + \text{SNR} + \text{INR} + 2\rho \sqrt{\SNRt \cdot \INRt}\nonumber \\
& - \frac{(\rho\SNRt + \sqrt{\SNRt\cdot\INRt})^2}{1 +\SNRt}\Bigg)     + 0.5\log\Sigma_{\ybold_2|\sbold} + C_G, \label{eq:zic-outer-gaussian6}
\end{align}
where $\Sigma_{\ybold_2|\sbold} \triangleq 1 + \SNRt - \Sigma_{\ybold_2,\sbold}\Sigma_{\sbold, \sbold}^{-1}\Sigma_{\ybold_2, \sbold}^T$, $\Sigma_{\ybold_2, \sbold} \triangleq \lsqb \rho\SNRt \:\:\: \rho\SNRt + \sqrt{\SNRt \cdot \INRt}\rsqb$ and $\Sigma_{\sbold, \sbold} \triangleq $\\
$\left[ \begin{array}{cc}
1 + \SNRt & \SNRt + \rho\sqrt{\SNRt\cdot\INRt} \\
\SNRt + \rho\sqrt{\SNRt\cdot\INRt} & 1 + \SNRt + \INRt + 2\rho\sqrt{\SNRt\cdot\INRt} \end{array} \right]$.
\end{theorem}
\begin{proof}
See Appendix~\ref{sec:append-ZIC-Gaussian-outer2}.
\end{proof}
\textit{Remarks:} 
\begin{itemize}
\item When there is no cooperation between the transmitters, the encoded messages at 
the two transmitters are independent of each other. Hence, for the 
non-cooperating case, the outer bound on the rate is obtained by setting $\rho = 0$ in Theorem~\ref{th:theorem-gauss-outer2}.
\item The outer bound in Theorem~\ref{th:theorem-gauss-outer2} is applicable over all the interference regimes. Note that the outer bound in Theorem~\ref{th:theorem-gauss-outer1} is also applicable to the high interference regime.  In the later part of the paper, it is demonstrated that the outer bound in Theorem~\ref{th:theorem-gauss-outer2} is tighter than the outer bound in Theorem~\ref{th:theorem-gauss-outer1} in this interference regime. 
\end{itemize}
\subsection{Relation between the outer bounds for the deterministic and Gaussian models}
In the following, it is shown that, for high $\SNRt$ and $\INRt$, the outer bounds for the Gaussian case in Theorems~\ref{th:theorem-gauss-outer1a}
and \ref{th:theorem-gauss-outer2} are approximately equal to the outer bounds for the deterministic model. For 
ease of presentation, it is assumed that $0.5 \log \SNRt$, $0.5 \log \INRt$, and $C_G$
are integers. Recall that, the parameters $m$, $n$ and $C$ of the deterministic model are related to the Gaussian
model as $m = (\lfloor 0.5 \log \text{SNR}\rfloor)^{+}$, $n = (\lfloor 0.5 \log \text{INR}\rfloor)^{+}$ and $C = \lfloor C_G \rfloor$, 
respectively.
\subsubsection{Weak/moderate interference regime $(0 \leq \alpha \leq 1)$}
It is easy to see that for high SNR and INR $(\text{i.e.}, \SNRt,~\INRt~\gg~1)$, the outer bounds on the individual 
rates in Theorem~\ref{th:theorem-gauss-outer1} can be approximated as 
\begin{align}
& R_1  \leq 0.5 \log(1 + \SNRt) \approx m, \nonumber \\
\text{ and } & R_2  \leq 0.5 \log(1 + \SNRt) \approx m. \label{eq:determin-gaussian1}
\end{align}
When $\SNRt > \INRt$ (i.e., $0 \leq \alpha \leq 1$), the outer bound on the sum rate  in Theorem~\ref{th:theorem-gauss-outer1} is approximated as
\begin{align}
R_1 + R_2 & \leq 0.5\log\lb 1 + \SNRt + \INRt + 2\sqrt{\SNRt \cdot \INRt} \rb \nonumber \\
& \qquad + 0.5\log\lb 1 + \frac{\SNRt}{1+\INRt}\rb + C_G, \nonumber \\
& \approx 2m-n + C.  \label{eq:determin-gaussian2}
\end{align}
From (\ref{eq:determin-gaussian1}) and (\ref{eq:determin-gaussian2}), the outer bound 
derived for the Gaussian case matches with the corresponding outer bound for the deterministic model stated in 
Theorem~\ref{th:theorem-weakmod-outer1}.

In Theorem~\ref{th:theorem-gauss-outer1a}, due to the maximization involved in the outer bound on $R_2$ over $\rho$, $C_G=0$ is considered to simplify the 
exposition. For the non-cooperating case, the outer bound is optimized by setting $\rho=0$. The outer bound on the rate of user~$2$ is approximated as
\begin{align}
R_2 & \leq 0.5\log \lb 1 + \SNRt - \frac{\SNRt \cdot \INRt}{1 + \SNRt + \INRt}\rb, \nonumber \\
& \approx m.  \label{eq:determin-gaussian2d} 
\end{align}
Hence, the outer bound on the rate of user~$2$ is approximately equal to $m$ for high $\SNRt$ and $\INRt$. 

It is also easy to see that, for high SNR and 
INR, the outer bound on the sum rate  in Theorem~\ref{th:theorem-gauss-outer1a} can be approximated as
\begin{align}
R_1 + R_2 & \approx 2m -n + C. \label{eq:determin-gaussian2c}
\end{align}
It can be noticed that the outer bound 
derived for the Gaussian case corresponds to the outer bound for the deterministic model stated in 
Theorem~\ref{th:theorem-weakmod-outer1}. It is interesting to note that both the outer 
bounds on the sum rate in Theorems~\ref{th:theorem-gauss-outer1} and \ref{th:theorem-gauss-outer1a} 
correspond to the outer bound for the deterministic model stated in 
Theorem~\ref{th:theorem-weakmod-outer1} for high SNR and INR. However, as mentioned earlier in the remark to Theorem~\ref{th:theorem-gauss-outer1a}, the 
outer bound in Theorem~\ref{th:theorem-gauss-outer1a} is tighter than
Theorems~\ref{th:theorem-gauss-outer1}. However, for high 
values of SNR and INR, the gap between these two outer bounds decreases and these two outer bounds are approximately equal to each other. 
\subsubsection{High interference regime $(1 < \alpha < 2)$}
In Theorem~\ref{th:theorem-gauss-outer2}, due to the maximization involved in the outer bounds on $R_2$ and $R_1 + R_2$ over $\rho$, $C_G=0$ is considered to simplify the 
exposition. For the non-cooperating case, the outer bound is optimized by setting $\rho=0$. First, the outer bound on the rate of user~$1$ is approximated as
\begin{align}
R_1 \leq & 0.5\log(1 + \SNRt) \approx m. \label{eq:determin-gaussian3}
\end{align}
The outer bound on the rate of user~$2$ is also approximated as
\begin{align}
R_2 & \leq 0.5\log\lb 1 + \SNRt - \frac{\SNRt\cdot\INRt}{1 + \SNRt + \INRt}\rb, \nonumber \\
& \approx 2m-n. \label{eq:determin-gaussian4}
\end{align}
The outer bound on the sum rate becomes
\begin{align}
R_1 + R_2 & \leq 0.5\log \lb1 + \SNRt + \INRt - \frac{\SNRt\cdot\INRt}{1 + \SNRt}\rb \nonumber \\
& + 0.5\log\Sigma_{\ybold_2|\sbold},
 \label{eq:determin-gaussian5}
\end{align}
where with some algebraic manipulation it can be shown that $\Sigma_{\ybold_2|\sbold} = 1 + \SNRt - \Sigma_{\ybold_2,\sbold}\Sigma_{\sbold,\sbold}^{-1}\Sigma_{\ybold_2,\sbold}^T \approx 1$. 
Hence, the sum rate outer bound in (\ref{eq:determin-gaussian5}) reduces to
\begin{align}
R_1 + R_2 \leq m. \label{eq:determin-gaussian7}
\end{align}
From (\ref{eq:determin-gaussian3}), (\ref{eq:determin-gaussian4}), and (\ref{eq:determin-gaussian7}), it can be
observed that the approximated outer bound of Gaussian case in Theorem~\ref{th:theorem-gauss-outer2} matches with the outer bound of deterministic case  in 
Theorem~\ref{theorem-high-outer1} for the high interference regime.

This validates that the approaches used in obtaining outer bounds in the two models are consistent with each other.
\section{Numerical Results and Discussion}\label{sec:results}
In the following sections, some numerical examples are presented for the deterministic and Gaussian cases,
to get insights into the system performance in different interference regimes.
\subsection{Deterministic Z-IC with unidirectional transmitter cooperation}
In Fig.~\ref{fig:capacity_region_m5n3}, the outer bound on the secrecy capacity region given in Theorem~\ref{th:theorem-weakmod-outer1} is plotted for $m=5$, $n=3$ and various values of $C$. The outer bound exactly matches with the lower bound on the secrecy capacity region for the corresponding values of $C$ in \cite{partha-isit-2015}. It is interesting to note that, without cooperation, and under the secrecy constraint at
receiver~$1$, when the rate of user~$2$ is upper bounded by $5$ bits per channel use (bpcu), the rate of user~$1$ is upper bounded by $2$ bpcu, and vice-versa. With further increase in the value of $C$, the outer bound on the sum rate in Theorem~\ref{th:theorem-weakmod-outer1} indicates that the sum rate performance may increase. For $C \geq 3$, the outer bound suggests that both the users may be able to achieve $5$ bpcu, and the achievable result in \cite{partha-isit-2015} establishes this is indeed the case.
\begin{figure}[t]
	\begin{center}
		\includegraphics[width=3.5in, height=3in]{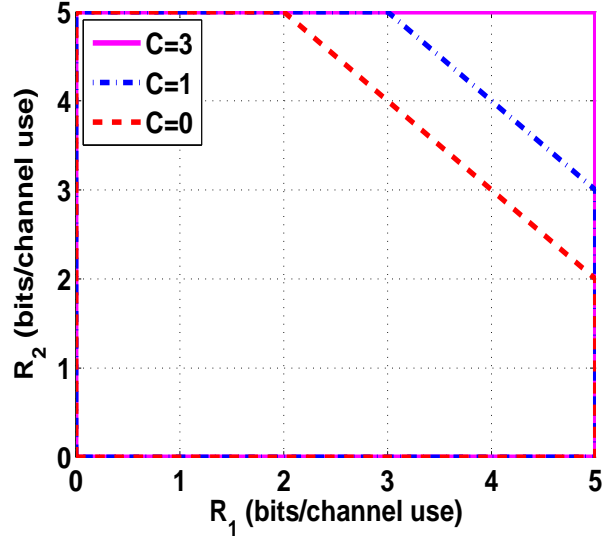}
		\caption{Secrecy capacity region of the deterministic Z-IC with~$(m,n)=(5,3)$. This corresponds to the moderate interference regime.}
		\label{fig:capacity_region_m5n3}
	\end{center}
\end{figure}
In Fig.~\ref{capacity_region_m4n5}, the outer bound on the secrecy capacity region given in Theorem~\ref{theorem-high-outer1} is 
plotted for $m=4$, $n=5$ and various values of $C$. When 
$C=0$ and the rate of user~$1$ is upper bounded by its maximum rate of $m$, i.e., $4$ bpcu, the outer bound establishes that user~$2$ cannot achieve any nonzero 
secrecy rate. When the rate of user~$2$ is upper bounded by $2m-n$, i.e., $3$ bpcu, the rate of user~$1$ is upper bounded by $1$ bpcu. When $C=1$, the outer bound on the sum rate in Theorem~\ref{theorem-high-outer1} suggests that both the users can achieve a nonzero secrecy rate with cooperation, in contrast to the non-cooperating case. The achievable result in \cite{partha-isit-2015} also confirms these observations and establishes the capacity region of the deterministic Z-IC with unidirectional transmitter cooperation and secrecy constraints at the receivers in the high interference regime.
\begin{figure}[t]
	\begin{center}
		\includegraphics[width=3.5in, height=3in]{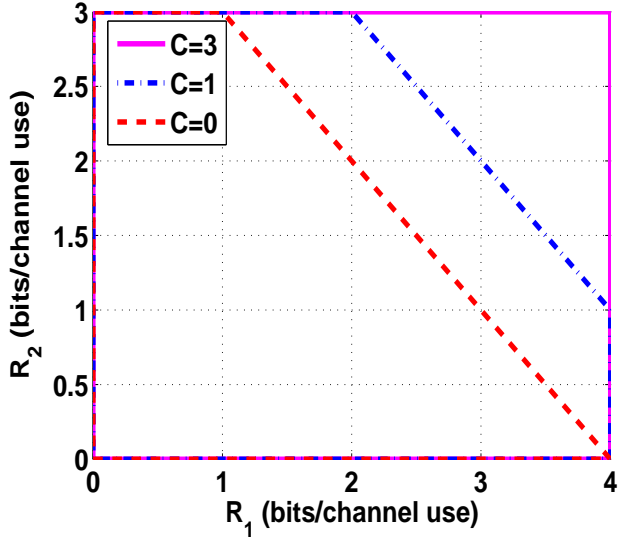}
		\caption{Secrecy capacity region of the deterministic Z-IC with~$(m,n)=(4,5)$. This corresponds to the high interference regime.}
		\label{capacity_region_m4n5}
	\end{center}
\end{figure}
\subsection{Gaussian Z-IC with unidirectional transmitter cooperation}
\begin{figure}[t]
	\begin{center}
		\includegraphics[width=3.5in, height=3in]{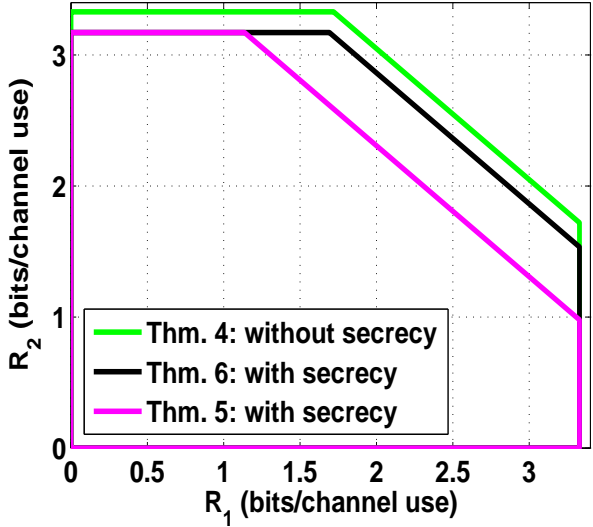}
		\caption{Comparison of the outer bounds on the secrecy capacity region for the Gaussian Z-IC: $P= 100$, $h_d = 1$, 
			$h_c = 0.5$ and $C_G= 0$.}
		\label{fig:outerbound_comp1}
	\end{center}
\end{figure}
In Fig.~\ref{fig:outerbound_comp1}, the outer bounds on the secrecy capacity 
region of the Z-IC in Theorems~\ref{th:theorem-gauss-outer1},~\ref{th:theorem-gauss-outer1a}~and~\ref{th:theorem-gauss-outer2} 
are compared for the weak/moderate interference regime. The outer bound in Theorem~\ref{th:theorem-gauss-outer1a}
is tight as compared to the outer bounds in Theorems~\ref{th:theorem-gauss-outer1} 
and \ref{th:theorem-gauss-outer2} except for the corner points for transmitter~$2$. Recall that, the outer bound in 
Theorem~\ref{th:theorem-gauss-outer1} 
does not use the secrecy constraint at the receiver in its derivation. The outer bound in Theorem~\ref{th:theorem-gauss-outer2}
is derived using the intuitions obtained from  the high interference
regime case considered in the deterministic model for
Theorem~\ref{theorem-high-outer1}. This is reflected in the plot as explained above. 
\begin{figure}[t]
	\begin{center}
		\includegraphics[width=3.5in, height=3in]{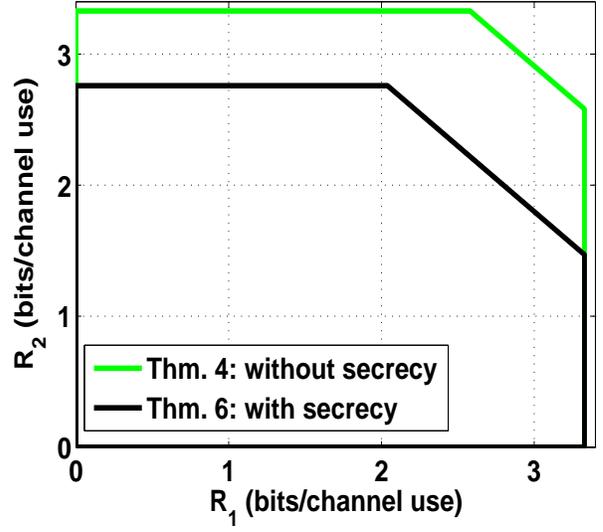}
		\caption{Comparison of the outer bounds on the secrecy capacity region for the Gaussian Z-IC: $P= 100$, $h_d = 1$, 
			$h_c = 1.5$ and $C_G= 1$.}
		\label{fig:outerbound_comp2}
	\end{center}
\end{figure}
In Fig.~\ref{fig:outerbound_comp2}, the outer bound on the secrecy capacity 
region of the Z-IC in Theorems~\ref{th:theorem-gauss-outer1}~and~\ref{th:theorem-gauss-outer2} 
are compared for the high interference regime. From the plot, it can be seen that the proposed outer bound
is tight as compared to the outer bound in 
Theorem~\ref{th:theorem-gauss-outer1}. 
\begin{figure}
	\centering
	\includegraphics[width=3.5in, height=3in]{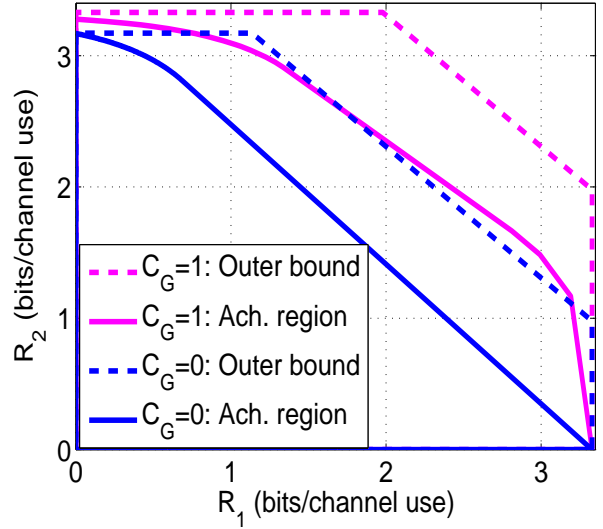}\\
	\caption{Comparison of the outer bounds with the achievable rate region for the Gaussian Z-IC: $P=100$, $h_d=1$ and $h_c=0.5$.}\label{fig:gaussian_result1}
\end{figure}
In Figs.~\ref{fig:gaussian_result1} and \ref{fig:gaussian_result2}, the outer bounds on the secrecy capacity region of the Gaussian Z-IC are plotted for different values of $C_G$ for the weak/moderate and high interference regimes, respectively. As the capacity of the cooperative link increases, the outer bounds indicate that the secrecy capacity region can enlarge in both the cases. This can also be observed from the lower bounds on the secrecy capacity region (curves labeled \texttt{Ach. region}) plotted in these figures using the result in \cite{partha-arxiv-2016, partha-inner-submitted-2016}.
\begin{figure}
	\centering 
	\includegraphics[width=3.5in, height=3in]{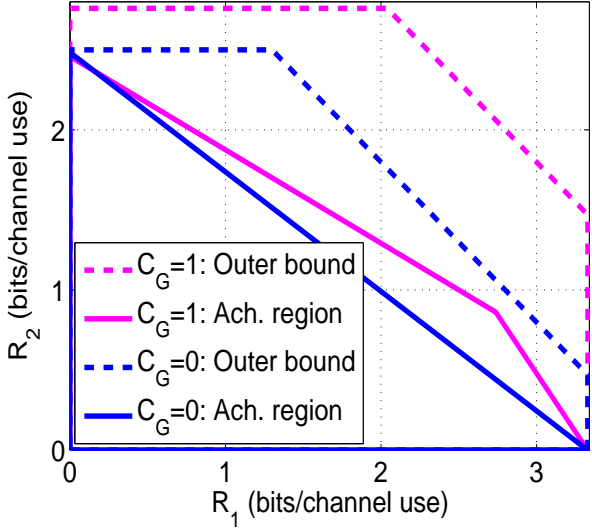}\\
	\caption{Comparison of the outer bounds with the achievable rate region for the Gaussian Z-IC: $P=100$, $h_d=1$ and $h_c=1.5$.}\label{fig:gaussian_result2}
\end{figure}
\section{Conclusions}\label{sec:conclusion}
This work derived outer bounds on the secrecy capacity region of the $2$-user Z-IC with limited-rate unidirectional transmitter cooperation. The outer bounds derived for the deterministic Z-IC model were shown to be tight for all the interference regimes and all possible values of $C$. One of the key techniques used in these derivations was to partition the encoded messages and outputs depending on the value of $\alpha$.  The outer bounds on the secrecy capacity region of the Gaussian Z-IC were derived using the insights obtained from the deterministic model. The outer bounds developed for the deterministic model helped to establish that secrecy can be obtained for free in the weak/moderate interference regime. However, the developed outer bounds suggest that there can be nonzero penalty on the rate of user~$2$ in all the interference regimes for the Gaussian case. The outer bounds also indicate that transmitter cooperation can help improve the performance of the system in the weak, moderate and high interference regimes, for both the models.
\appendix
\subsection{Proof of Theorem~\ref{th:theorem-weakmod-outer1}}\label{sec:append-det-ZIC-outer1}
In the Z-IC model considered in this paper,  there is 
unidirectional cooperation from transmitter~$2$ to transmitter~$1$. Due to this, neither transmitter can aid
in relaying other transmitter's message. Thus, a trivial outer bound on the individual rate of each user is $m$. Hence, it is
only required to establish the bound on the sum rate. Starting from Fano's inequality, the proof goes as follows. 
Receiver~$2$ is provided with side information $(\xbold_{2a}^N,\vtwoone^N)$ by a genie, which leads to further
upper bounding the sum rate. Providing
this side information helps to cancel the negative entropy term in the outer bound for the sum rate,
leading to a tractable outer bound. 
\begin{align}
& N[R_1 + R_2] \nonumber \\
& \leq I(W_1;\ybold_1^N) + I(W_2;\ybold_2^N) + N\epsilon_N, \nonumber \\
& \qquad \qquad \quad \text{ where } \epsilon_N \rightarrow 0 \text{ as }
N \rightarrow \infty, \nonumber \\
&\leq I(W_1;\ybold_1^N) + I(W_2;\ybold_2^N,\xbold_{2a}^N,\vtwoone^N) + N\epsilon_N, \nonumber 
\end{align}
\begin{align}
& = H(\ybold_1^N) - H(\ybold_1^N|W_1) + H(\vtwoone^N) - H(\vtwoone^N|W_2) \nonumber \\
& \quad + H(\xbold_{2a}^N|\vtwoone^N)  -  H(\xbold_{2a}^N|\vtwoone^N,W_2) + H(\ybold_2^N|\xbold_{2a}^N,\vtwoone^N) \nonumber \\
& \quad - H(\ybold_2^N|\xbold_{2a}^N,\vtwoone^N,W_2) + N\epsilon_N, \nonumber 
\\
& \leq H(\ybold_1^N) - H(\ybold_1^N|W_1, \xbold_1^N,\vtwoone^N) + H(\vtwoone^N) - H(\vtwoone^N|W_2)  \nonumber \\
& \quad + H(\xbold_{2a}^N|\vtwoone^N) -  H(\xbold_{2a}^N|\vtwoone^N,W_2)  + H(\xbold_{2b}^N|\xbold_{2a}^N,\vtwoone^t) \nonumber \\
&\quad - H(\xbold_{2b}^N|\xbold_{2a}^N,\vtwoone^N,W_2) + N\epsilon_N, 
\end{align}
where the last step is obtained using the fact that conditioning cannot increase the 
entropy. 

Note that due to cooperation between the transmitters, the encoded messages are dependent, and hence, it is difficult to 
bound or simplify the entropy terms. Here, partitioning of the output $\ybold_1 = (\xbold_{1a},\xbold_{1b} \oplus \xbold_{2a})$
as shown in Fig.~\ref{fig:weakmodouter} helps to simplify the bound further. In 
the following, the fact that removing conditioning cannot decrease the entropy 
has also been used. 
\begin{align}
& N [R_1 + R_2]  \nonumber \\
& \stackrel{(a)}{\leq} H(\ybold_1^N) - H(\xbold_{2a}^N|\vtwoone^N) + H(\vtwoone^N) - H(\vtwoone^N|W_2)  \nonumber \\
& \quad + H(\xbold_{2a}^N|\vtwoone^N) -  H(\xbold_{2a}^N|\vtwoone^N,W_2)  + H(\xbold_{2b}^N) \nonumber \\
&\quad- H(\xbold_{2b}^N|\xbold_{2a}^N,\vtwoone^N,W_2) + N\epsilon_N, \nonumber \\
& \leq H(\ybold_1^N) + H(\vtwoone^N) + H(\xbold_{2b}^N) + N\epsilon_N, \nonumber \\
&\text{or }  R_1 + R_2   \stackrel{(b)}{\leq} 2m-n + C,
\end{align}
where (a) is obtained using the relation in (\ref{eq:usefulrelation}) and (b) is obtained by bounding the entropy
terms $H(\ybold_1)$, $ H(\vtwoone)$ and $H(\xbold_{2b})$ by $m$, $C$ and $m-n$, 
respectively. This completes the proof.
\subsection{Proof of Theorem~\ref{theorem-high-outer1}}\label{sec:append-ZIC-outer2}
As mentioned earlier in the proof of Theorem~\ref{th:theorem-weakmod-outer1}, the  rate of user~$1$ is upper bounded by $m$. The output at receiver~$1$ is partitioned into two parts $\ybold_1^N = (\ybold_{1a}^N, \ybold_{1b}^N)$ as shown in Fig.~\ref{fig:highouter}. 
A genie provides part of 
the output  of receiver~$1$, namely, $\ybold_{1a}^N$, as side-information to receiver~$2$. Note that, $\ybold_{1a}^N$ does not contain any signal sent from transmitter~$1$. Using Fano's inequality the following is obtained
\begin{align}
NR_2  \leq I(W_2;\ybold_{1a}^N) + I(W_2;\ybold_2^N|\ybold_{1a}^N) + N\epsilon_N, \label{eq:th-high-outer2}
\end{align}
Using the secrecy constraint at receiver~$1$, $I(W_2; \ybold_1^N) = I(W_2; \ybold_{1a}^N, \ybold_{1b}^N) \leq N\epsilon_N$, the first term above is upper bounded as $I(W_2;\ybold_{1a}^N) \leq N\epsilon_N$ as mutual information cannot be negative.  Hence, $NR_2 \leq H(\ybold_2^N|\ybold_{1a}^N) + N\epsilon_N$, which can be further upper bounded as $NR_2 \leq H(\xbold_{2b}^N|\xbold_{2a}^N) + N\epsilon_N$. Since, $H(\xbold_{2b})$ is upper bounded by $2m-n$, one obtains $R_2 \leq 2m-n$.

Next, using Fano's inequality and providing $\ybold_{1a}^N$ as side information to receiver~$2$, the sum rate is upper bounded as
\begin{align}
N[R_1 + R_2] \leq I(W_1;\ybold_1^N) + I(W_2;\ybold_2^N, \ybold_{1a}^N) + 
N\epsilon_N. \label{eq:th-high-outer2d}
\end{align}
Using the fact that the encoding at transmitter~$2$ does not depend on $W_1$, it can be seen that $I(W_1;\ybold_{1a}^N,\ybold_{1b}^N) = 
I(W_1;\ybold_{1b}^N|\xbold_{2a}^N)$. The second mutual information term in 
(\ref{eq:th-high-outer2d}) is upper bounded using the relation $I(W_2;\ybold_{1a}^N) \leq N\epsilon_N$. Hence, the outer bound on the sum rate becomes
\begin{align}
& N[R_1 + R_2] \nonumber \\
& \leq  I(W_1;\ybold_{1b}^N|\xbold_{2a}^N) + I(W_2;\ybold_2^N|\ybold_{1a}^N) + N\epsilon_N, \nonumber \\
& \leq H(\ybold_{1b}^N|\xbold_{2a}^N) - H(\ybold_{1b}^N|\xbold_{2a}^N,W_1) + H(\ybold_2^N|\ybold_{1a}^N) \nonumber \\
& \quad -H(\ybold_2^N|\ybold_{1a}^N,W_2) + N\epsilon_N, \nonumber \\
& \stackrel{(a)}{\leq} H(\ybold_{1b}^N|\xbold_{2a}^N) - H(\ybold_{1b}^N|\xbold_{2a}^N,W_1,\xbold_1^N) + H(\ybold_2^N|\ybold_{1a}^N) \nonumber \\
&\quad -H(\ybold_2^N|\ybold_{1a}^N,W_2) + N\epsilon_N, \nonumber \\
& \leq H(\ybold_{1b}^N|\xbold_{2a}^N) - H(\xbold_{2b}^N,\xbold_{2c}^N|\xbold_{2a}^N,W_1,\xbold_1^N,\vtwoone^N)+ H(\xbold_{2b}^N|\xbold_{2a}^N) \nonumber \\
&\quad - H(\xbold_{2b}^N|\xbold_{2a}^N,W_2) + N\epsilon_N, \nonumber 
\end{align}
\begin{align}
& \stackrel{(b)}{=} H(\ybold_{1b}^N|\xbold_{2a}^N) - H(\xbold_{2b}^N,\xbold_{2c}^N|\xbold_{2a}^N,\vtwoone^N)+ H(\xbold_{2b}^N|\xbold_{2a}^N)  \nonumber \\
& \quad - H(\xbold_{2b}^N|\xbold_{2a}^N,W_2) + N\epsilon_N, \nonumber \\
& \leq H(\ybold_{1b}^N) - H(\xbold_{2b}^N|\xbold_{2a}^N,\vtwoone^N) - H(\xbold_{2c}^N|\xbold_{2b}^N,\xbold_{2a}^N,\vtwoone^N) \nonumber \\
& \quad + H(\xbold_{2b}^N,\vtwoone^N|\xbold_{2a}^N)  - H(\xbold_{2b}^N|\xbold_{2a}^N,W_2) + N\epsilon_N, \nonumber \\
& \stackrel{(c)}{\leq} H(\ybold_{1b}^N) - H(\xbold_{2b}^N|\xbold_{2a}^N,\vtwoone^N) - H(\xbold_{2c}^N|\xbold_{2b}^N,\xbold_{2a}^N,\vtwoone^N)  \nonumber \\
& \quad + H(\vtwoone^N) + H(\xbold_{2b}^N|\xbold_{2a}^N,\vtwoone^N) - H(\xbold_{2b}^N|\xbold_{2a}^N,W_2) + N\epsilon_N, \nonumber \\
& \leq H(\ybold_{1b}^N) + H(\vtwoone^N) + N\epsilon_N, \nonumber \\
& \text{or } R_1 + R_2  \leq m + C, \label{eq:th-high-outer3}
\end{align}
where (a) is because conditioning cannot increase the
 entropy; (b) is obtained using the relation in (\ref{eq:usefulrelation}); (c) follows because removing conditioning
 cannot decrease the entropy, and using the chain rule for joint entropy. This completes the proof.
\subsection{Proof of Theorem~\ref{theorem-veryhigh-outer1}}\label{sec:append-ZIC-outer3}
As mentioned in the proof of Theorem~\ref{th:theorem-weakmod-outer1}, the rate of user~$1$ is upper bounded
by $m$. To bound the rate of user~$2$, $\ybold_{1a}^N$ is provided as side-information
to receiver~$2$ as follows and the outer bound is simplified as follows
\begin{align}
NR_2 & \leq I(W_2;\ybold_2^N,\ybold_{1a}^N) + N\epsilon_N, \nonumber \\
& = I(W_2;\ybold_{1a}^N) + I(W_2;\ybold_2^N|\ybold_{1a}^N)+ N\epsilon_N, \nonumber \\
\text{or } R_2 \leq 0,  \label{eq:th-veryhigh-outer2}
\end{align}
where the above equation is obtained using the secrecy constraint at receiver~$1$, i.e., $I(W_2;\ybold_{1a}^N)\leq N\epsilon$ and $I(W_2;\ybold_2^N|\ybold_{1a}^N)= 0$ as 
observed from Fig.~\ref{fig:veryhighouter}. This completes the proof. 
\subsection{Proof of Theorem~\ref{th:theorem-gauss-outer1a}}\label{sec:append-ZIC-Gaussian-outer1a} 
It is easy to see that the rate of transmitter~$1$ is upper bounded by $0.5\log(1+\text{SNR})$. Hence, it is
required to proof the outer bounds on the rate of transmitter~$2$ and the sum rate. Using Fano's inequality,
rate of transmitter~$2$ is upper bounded as follows
\begin{align}
NR_2 & \leq I(W_2;\ybold_2^N) + N\epsilon_N, \nonumber \\
& \leq I(W_2;\ybold_2^N,\ybold_1^N) + N\epsilon_N, \nonumber \\
& = I(W_2;\ybold_1^N)+ I(W_2;\ybold_2^N|\ybold_1^N) + N\epsilon_N, \nonumber \\
& \stackrel{(a)}{\leq} h(\ybold_2^N|\ybold_1^N) - h(\ybold_2^N|\ybold_1^N,W_2) + N\epsilon_N, \nonumber \\
\text{or } R_2 & \stackrel{(b)}{\leq} \displaystyle\max_{0 \leq |\rho| \leq 1}0.5\log\Bigg( 1 + \text{SNR} - \nonumber \\
& \quad  \frac{(\rho\text{SNR} + \sqrt{\text{SNR}\cdot\text{INR}})^2}{1+ \text{SNR}+\text{INR}+2\rho\sqrt{\text{SNR}\cdot\text{INR}}}\Bigg), \label{eq:outertwo-weak-gaussian1}
\end{align}
where (a) is obtained using the secrecy constraint at the receiver~$1$; (b) is obtained using the fact that for
a given power constraint, the differential entropy is maximized by the Gaussian distribution.

In the following, sum rate is upper bounded using  Fano's inequality, secrecy 
constraint at receiver~$1$ and chain rule of mutual information. 
\begin{align}
& N[R_1 + R_2] \nonumber \\
& \leq I(W_1; \ybold_1^N) + I(W_2; \ybold_2^N) - I(W_2; \ybold_1^N) + N\epsilon_N, \nonumber \\
& = I(W_1; \ybold_1^N) + I(W_2; \ybold_2^N) - I(W_2;\ybold_1^N,\sbold_2^N)  \nonumber \\
& \qquad + I(W_2;\sbold_2^N|\ybold_1^N) + N\epsilon_N, \text{ where } \sbold_2^N \triangleq h_c \xbold_2^N + \zbold_1^N. \label{eq:outer-weakmod-append1}
\end{align}
The main novelty in the proof lies in  bounding these mutual information 
terms. To upper bound the sum rate further, consider the first two terms of 
(\ref{eq:outer-weakmod-append1}), where the cooperative signal $\vtwoone^N$ is provided as 
side-information to both the receivers. 
\begin{align}
& I(W_1;\ybold_1^N) + I(W_2;\sbold_2^N|\ybold_1^N) \nonumber \\
& \stackrel{(a)}{\leq} I(W_1;\ybold_1^N|\vtwoone^N) + I(W_2;\vtwoone^N|\ybold_1^N) + I(W_2;\sbold_2^N|\ybold_1^N,\vtwoone^N), \nonumber \\
& \leq I(W_1, \xbold_1^N;\ybold_1^N|\vtwoone^N) + I(W_2;\vtwoone^N|\ybold_1^N) + I(W_2;\sbold_2^N|\ybold_1^N,\vtwoone^N), \nonumber
\\
& \stackrel{(b)}{=} I(\xbold_1^N;\ybold_1^N|\vtwoone^N) + I(W_2;\vtwoone^N|\ybold_1^N) + I(W_2;\sbold_2^N|\ybold_1^N,\vtwoone^N), \nonumber \\
& = I(\xbold_1^N;\ybold_1^N|\vtwoone^N) + H(\vtwoone^N|\ybold_1^N) - H(\vtwoone^N|\ybold_1^N,W_2) \nonumber \\
& \quad + h(\sbold_2^N|\ybold_1^N,\vtwoone^N) - h(\sbold_2^N|\ybold_1^N,\vtwoone^N, W_2), \nonumber
\\
& \stackrel{(c)}{\leq} I(\xbold_1^N;\ybold_1^N|\vtwoone^N) +  H(\vtwoone^N) + h(\sbold_2^N, \ybold_1^N|\vtwoone^N) - h(\ybold_1^N|\vtwoone^N) \nonumber \\
& \quad- h(\sbold_2^N|\ybold_1^N,\vtwoone^N,W_2) , \nonumber
\\
& = I(\xbold_1^N;\ybold_1^N|\vtwoone^N) + H(\vtwoone^N)+ h(\sbold_2^N|\vtwoone^N) + h(\ybold_1^N|\sbold_2^N, \vtwoone^N) \nonumber \\
& - h(\ybold_1^N|\vtwoone^N) - h(\sbold_2^N|\ybold_1^N,\vtwoone^N,W_2), \nonumber \\
\end{align}
where (a) is obtained using the chain rule for mutual information and the fact that $\vtwoone$ is not a function of $W_1$;
 (b) is obtained using the Markov chain relation: $W_1 \rightarrow (\vtwoone, \xbold_1)  \rightarrow \ybold_1$,  
 which can shown using the signal flow graph (SFG) approach in \cite{kramer2008topics};
 (c) follows because removing conditioning cannot decrease entropy and
 $h(\sbold_2^N, \ybold_1^N|\vtwoone^N) = h(\ybold_1^N|\vtwoone^N) + h(\sbold_2^N|\ybold_1^N, 
 \vtwoone^N)$. 
 
Note that the bounding these differential entropy terms in above is difficult as it 
involves continuous and discrete random variables. To overcome this problem, 
using relation in (\ref{eq:usefulrelation}), it can be shown that $h(\sbold_2^N|\vtwoone^N) = h(\sbold_2^N|\vtwoone^N, 
\xbold_1^N)$. This also implies that $h(\sbold_2^N|\vtwoone^N, 
\xbold_1^N) = h(\ybold_1^N|\vtwoone^N, \xbold_1^N)$. This  is one of the key 
steps in the derivation as it leads to cancelation of $I(\xbold_1^N;\ybold_1^N|\vtwoone^N)$ 
as shown below.
\begin{align}
& I(W_1;\ybold_1^N) + I(W_2;\sbold_2^N|\ybold_1^N) \nonumber \\
& \leq I(\xbold_1^N;\ybold_1^N|\vtwoone^N) + H(\vtwoone^N) + h(\sbold_2^N|\vtwoone^N, \xbold_1^N)  \nonumber \\
& \quad + h(\ybold_1^N|\sbold_2^N, \vtwoone^N) - h(\ybold_1^N|\vtwoone^N) - h(\sbold_2^N|\ybold_1^N,\vtwoone^N,W_2), \nonumber 
\\
& = I(\xbold_1^N;\ybold_1^N|\vtwoone^N) + H(\vtwoone^N) + h(\ybold_1^N|\vtwoone^N, \xbold_1^N) \nonumber \\
& \quad + h(\ybold_1^N|\sbold_2^N, \vtwoone^N)  - h(\ybold_1^N|\vtwoone^N) - h(\sbold_2^N|\ybold_1^N,\vtwoone^N,W_2), \nonumber \\
& \stackrel{(a)}{\leq}  I(\xbold_1^N;\ybold_1^N|\vtwoone^N) + NC_G  - I(\xbold_1^N;\ybold_1^N|\vtwoone^N)  + h(\ybold_1^N|\sbold_2^N, \vtwoone^N)  \nonumber \\
& \quad- h(\sbold_2^N|\ybold_1^N,\vtwoone^N, W_2,\xbold_2^N), \nonumber \\
& \stackrel{(b)}{=} NC_G + h(\ybold_1^N|\sbold_2^N, \vtwoone^N)  - h(\sbold_2^N|\ybold_1^N,\vtwoone^N, \xbold_2^N), \nonumber \\
&= NC_G +  h(\ybold_1^N|\sbold_2^N, \vtwoone^N)  -h(\sbold_2^N, \ybold_1^N|\vtwoone^N, \xbold_2^N) \nonumber \\
& \quad + h(\ybold_1^N|\vtwoone^N, \xbold_2^N), \nonumber 
\end{align}
\begin{align}
& = h(\ybold_1^N|\sbold_2^N, \vtwoone^N) - h(\sbold_2^N|\xbold_2^N, \vtwoone^N) - h(\ybold_1^N|\sbold_2^N, \xbold_2^N, \vtwoone^N) \nonumber \\
& \quad + h(\ybold_1^N|\xbold_2^N, \vtwoone^N) + NC_G, \nonumber \\
& \stackrel{(c)}{\leq} h(\sbold_1^N) - h(\zbold_1^N) + N C_G, \text{ where } \sbold_1^N \triangleq h_d\xbold_1^N + \zbold_1^N, \label{eq:outer-weakmod-append2}
\end{align}
where (a) is obtained using the fact that conditioning cannot
 increase the differential entropy and $H(\vtwoone^N) \leq NC_G$; (b) is obtained using the fact that
 $I(W_2;\sbold_2^N|\ybold_1^N,\vtwoone^N, \xbold_2^N) = 0$, which can again be shown with the help of an 
 SFG~\cite{kramer2008topics}; and (c) is obtained by noticing that first and third term cancel with each other
 using the relation in (\ref{eq:usefulrelation}) and using the fact
  that conditioning cannot increase the differential entropy.

Now, consider the bounding of the remaining two terms in (\ref{eq:outer-weakmod-append1}). As it involves difference of
two mutual information terms, it is not straightforward to upper bound these terms. In the weak/moderate interference
regime, the channel from transmitter~$2$ to receiver~$1$ is weaker as compared to the channel from transmitter~$2$ to
receiver~$1$. Hence, $\xbold_2$, $\ybold_2$ and $\sbold_2$ satisfy the following Markov chain: 
$\xbold_2 \rightarrow \ybold_2 \rightarrow \sbold_2$ and this channel can be viewed as a degraded broadcast channel
(BC). Using the result in \cite{li-isit-2008, cheong-TIT-1978}, following bound is obtained. 
\begin{align}
& I(W_2;\ybold_2^N) - I(W_2;\ybold_1^N,\sbold_2^N) \nonumber \\
& = I(W_2;\ybold_2^N) - I(W_2;\sbold_2^N)- I(W_2,\ybold_1^N|\sbold_2^N), \nonumber \\
& \leq I(W_2;\ybold_2^N) - I(W_2;\sbold_2^N), \nonumber \\
& \leq N[I(\xbold_2;\ybold_2) - I(\xbold_2;\sbold_2)], \label{eq:outer-weakmod-append3}
\end{align}

Finally, using (\ref{eq:outer-weakmod-append2}) and (\ref{eq:outer-weakmod-append3}),
(\ref{eq:outer-weakmod-append1}) becomes
\begin{align}
R_1 + R_2 & \leq \log(1 + \SNRt) - 0.5\log(1 + \INRt) + C_G, \label{eq:outer-weakmod-append4}
\end{align}
where the above equation is obtained using the fact that for a given power constraint, Gaussian distribution 
maximizes the differential entropy. This completes the proof. 
\subsection{Proof of Theorem~\ref{th:theorem-gauss-outer2}}\label{sec:append-ZIC-Gaussian-outer2}
As mentioned earlier, rate of transmitter~$1$ is upper bounded by $0.5\log(1+\text{SNR})$. Hence, it is
required to proof the outer bounds on the rate of transmitter~$2$ and the sum rate. Using the steps used to obtain outer bound on the rate of user~$2$ in the proof of Theorem~\ref{th:theorem-gauss-outer1a}, following bound is obtained
\begin{align}
NR_2 &\leq \displaystyle\max_{0 \leq |\rho| \leq 1}0.5\log\Bigg( 1 + \text{SNR} \nonumber \\
& \quad - \frac{(\rho\text{SNR} + \sqrt{\text{SNR}\cdot\text{INR}})^2}{1+ \text{SNR}+\text{INR}+2\rho\sqrt{\text{SNR}\cdot\text{INR}}}\Bigg), \label{eq:outertwo-gaussian1}
\end{align}

The derivation of the outer bound on the sum rate goes as follows. First, an 
outer bound on the rate of user~$1$ is obtained. Then, an outer bound on the 
rate of user~$2$ is derived. Adding these two outer bounds leads to cancelation of 
negative differential entropy terms, which in turn allows to obtain single 
letter characterization of the sum rate outer bound. 

In the following, an outer bound on the rate of user~$1$ is obtained providing $\ybold_2^N$ as side-information to receiver~$1$. 
\begin{align}
NR_1 &  \leq I(W_1;\ybold_1^N, \ybold_2^N) + N\epsilon_N, \nonumber \\
& \stackrel{(a)}{=} I(W_1;\ybold_1^N|\ybold_2^N) + N\epsilon_N, \nonumber \\
& \stackrel{(b)}{\leq}  h(\ybold_1^N|\ybold_2^N) - h(\sbold_1^N|\ybold_2^N, W_1, \xbold_2^N, \vtwoone^N) + N\epsilon_N, \nonumber \\
& \qquad \quad \text{ where } \sbold_1^N \triangleq h_d \xbold_1^N + \zbold_1^N\nonumber \\
& \stackrel{(c)}{\leq}  h(\ybold_1^N|\ybold_2^N) - h(\tilde{\sbold}_1^N|\ybold_2^N, W_1, \xbold_2^N, \vtwoone^N) + N\epsilon_N, \nonumber \\
& \qquad \quad \text{ where } \tilde{\sbold}_1^N \triangleq h_d \xbold_1^N + \tilde{\zbold}_1^N, \nonumber \\
&  \stackrel{(d)}{=}  h(\ybold_1^N|\ybold_2^N) - h(\tilde{\sbold}_1^N|W_1, \vtwoone^N) + N\epsilon_N, \label{eq:outertwo-gaussian2}
\end{align}
where (a) is obtained
using the fact that $\ybold_2^N$ is independent of $W_1$; (b) is obtained using the fact that
conditioning cannot increase the differential entropy; (c) is obtained using the fact that the secrecy capacity
region of Z-IC with confidential messages is invariant under any joint channel noise distribution $P(\zbold_1^N,\zbold_2^N)$
that leads to the same marginal distributions $P(\zbold_1^N)$ and $P(\zbold_2^N)$ \cite{he-CISS-2009}. Although this
invariance property is stated for the Gaussian IC in \cite{he-CISS-2009}, it holds for the Z-IC with limited-rate 
transmitter cooperation also. The need for replacing $\zbold_1^N$ with $\tilde{\zbold}_1^N$ will become clear 
later in the proof. Finally,  (d) is obtained using the relation
 in (\ref{eq:usefulrelation}).

Next, to bound the rate of user~$2$, starting from Fano's inequality, one proceeds as follows. 
The genie provides $(\ybold_1^N, W_1)$ as side-information to receiver~$2$ and 
the sum rate is further upper bounded as follows
\begin{align}
NR_2 \leq I(W_2; \ybold_1^N, W_1) + I(W_2; \ybold_2^N | \ybold_1^N, W_1) + N\epsilon_N. \label{eq:outertwo-gaussian3}
\end{align}
Consider the first term in (\ref{eq:outertwo-gaussian3})
\begin{align}
I(W_2; \ybold_1^N, W_1) & \stackrel{(a)}{\leq} N\epsilon_N + H(W_1|\ybold_1^N) - H(W_1|\ybold_1^N,W_2), \nonumber \\
& \stackrel{(b)}{\leq} N\epsilon_N, \label{eq:outertwo-gaussian4}
\end{align}
where (a) is obtained using the secrecy constraint at receiver~$1$, i.e., $I(W_2;\ybold_1^N) \leq N\epsilon_N$
and (b) is obtained from the reliability condition for message $W_1$, i.e., $H(W_1|\ybold_1^N) \leq N \delta_N$ and dropping the negative entropy term. 
In above, for notational simplicity, $\delta_N$ is absorbed to $\epsilon_N$. 
Using (\ref{eq:outertwo-gaussian4}), (\ref{eq:outertwo-gaussian3}) reduces to 
\begin{align}
NR_2 & \leq I(W_2; \ybold_2^N, \vtwoone^N | \ybold_1^N, W_1) + N\epsilon_N, \nonumber \\
& = I(W_2; \vtwoone^N | \ybold_1^N, W_1) + I(W_2; \ybold_2^N |\vtwoone^N, \ybold_1^N, W_1)+ N\epsilon_N.
\end{align}
To bound the rate of user~$2$ further, $\tilde{\sbold}_1^N$ is included in the 
second mutual information term. In the following, it can be noticed that 
working with $\tilde{\sbold}_1^N$ instead of $\sbold_1^N$ leads to $-h(\tilde{\zbold}_1^N)$ 
instead of $0$. Thus, replacing the noise in $\sbold_1^N$ with an independent noise  leads to a tighter 
outer bound. Hence, the outer bound on $R_2$ becomes
\begin{align}
R_2 & \leq  H(\vtwoone^N|\ybold_1^N, W_1) - H(\vtwoone^N |\ybold_1^N, W_1, W_2) \nonumber \\
& \quad + I(W_2; \ybold_2^N, \tilde{\sbold}_1^N |\vtwoone^N, \ybold_1^N, W_1)+ N\epsilon_N, \nonumber \\
& \stackrel{(a)}{\leq} H(\vtwoone^N) +  I(W_2; \tilde{\sbold}_1^N |\vtwoone^N, \ybold_1^N, W_1) \nonumber \\
& \quad + I(W_2; \ybold_2^N |\vtwoone^N, \ybold_1^N, W_1,  \tilde{\sbold}_1^N) + N\epsilon_N, \nonumber \\
& \stackrel{(b)}{\leq} H(\vtwoone^N)+  h(\tilde{\sbold}_1^N |\vtwoone^N, W_1) - h(\tilde{\sbold}_1^N |\vtwoone^N, \ybold_1^N, W_1, W_2, \xbold_2^N) \nonumber \\
& \quad + h(\ybold_2^N |\ybold_1^N, \tilde{\sbold}_1^N) -  h(\ybold_2^N |\vtwoone^N, \ybold_1^N, W_1,\tilde{\sbold}_1^N, W_2, \xbold_2^N) \nonumber \\
& \quad + N\epsilon_N, \nonumber \\
& = H(\vtwoone^N) +  h(\tilde{\sbold}_1^N |\vtwoone^N, W_1) - h(\tilde{\zbold}_1^N) + h(\ybold_2^N |\ybold_1^N, \tilde{\sbold}_1^N) \nonumber \\
& \quad - h(\zbold_2^N) + N\epsilon_N, \label{eq:outertwo-gaussian5}
\end{align}
where (a) and (b) are obtained using the fact that
removing (or adding) conditioning cannot decrease (or cannot increase) the differential 
entropy.

Adding (\ref{eq:outertwo-gaussian2}) and (\ref{eq:outertwo-gaussian5}), the following is obtained
\begin{align}
 & R_1 + R_2  \nonumber \\ 
 & \leq H(\vtwoone) + h(\ybold_1|\ybold_2) + h(\ybold_2|\ybold_1, \tilde{\sbold}_1) - h(\tilde{\zbold}_1)  -h(\zbold_2), \nonumber \\
 & \leq \displaystyle\max_{0 \leq |\rho| \leq 1} C_G + 0.5\log\Bigg[ 1 + \text{SNR} + \text{INR} + 2\rho\sqrt{\SNRt \cdot \INRt} \nonumber \\
 & \quad- \frac{(\rho\SNRt + \sqrt{\SNRt\cdot\INRt})^2}{1 +\SNRt}\Bigg]   + 0.5\log\Sigma_{y_2|s},  \label{eq:outertwo-gaussian6}
\end{align}
where $\Sigma_{\ybold_2|\sbold}$ is as defined in the statement of the theorem.
The above equation is obtained using the fact that for a given power constraint, the Gaussian distribution maximizes
the conditional differential entropy. The individual terms in the above equations are simplified as follows
\begin{align}
h(\ybold_1|\ybold_2)& = 0.5\log 2\pi e \Sigma_{\ybold_1|\ybold_2}, \label{eq:outertwo-gaussian7} 
\end{align}
where 
\begin{align}
 \Sigma_{\ybold_1|\ybold_2}  & = E[\ybold_1^2] - \frac{E[\ybold_1\ybold_2]^2}{E[\ybold_2^2]}, \nonumber \\
& =  1 + \text{SNR} + \text{INR} + 2\rho\sqrt{\SNRt \cdot \INRt} \nonumber \\
& \quad - \frac{(\rho\SNRt + \sqrt{\SNRt\cdot\INRt})^2}{1 +\SNRt}. \label{eq:outertwo-gaussian8}
\end{align}
The term $\Sigma_{\ybold_2|\sbold}$ is obtained as follows
\begin{align}
\Sigma_{\ybold_2|\sbold} & = E[\ybold_2^2] - E[\ybold_2 \sbold^T] E[\sbold \sbold^T]^{-1} E[\sbold \ybold_2], \nonumber \\
& \qquad \qquad \quad \text{ where } \sbold \triangleq [\tilde{\sbold}_1\:\: \ybold_1]^T, \nonumber \\
& = 1 + \SNRt - \Sigma_{\ybold_2,\sbold}\Sigma_{\sbold, \sbold}^{-1}\Sigma_{\ybold_2, \sbold}^T. \label{eq:outertwo-gaussian9}
\end{align}
In the above equation, the terms $\Sigma_{\ybold_2,\sbold}$ and $\Sigma_{\sbold, \sbold}$ are as defined in the statement of the
theorem. This completes the proof.
\bibliographystyle{IEEEtran}
\bibliography{IEEEabrv,refs_onesidedIC}
\end{document}